\documentclass[twocolumn]{revtex4-1}
\usepackage{mathrsfs}
\usepackage{physics}
\usepackage{tikz}
\usepackage[caption=false]{subfig}
\usepackage[colorlinks=true,linkcolor=blue]{hyperref}
\usepackage{cleveref}
\usepackage{amsmath,amssymb,amsfonts,amsthm}
\usepackage{graphicx}
\usepackage{ragged2e}
\newtheorem{lem}{Lemma}
\begin{document}
\date{\today}
\author{Nikolaos K. Kollas}
\email{kollas@upatras.gr}
\affiliation{Division of Theoretical and Mathematical Physics, Astronomy and Astrophysics, Department of Physics, University of Patras, 26504,  Patras, Greece}
\author{Kostas Blekos}
\affiliation{Division of Theoretical and Mathematical Physics, Astronomy and Astrophysics, Department of Physics, University of Patras, 26504,  Patras, Greece}
\title{Faithful extraction of quantum coherence}
\begin{abstract}
Coherence, a strictly quantum phenomenon, has found many applications, from quantum information theory and thermodynamics to quantum foundations and biology. When physical constraints are taken into consideration creation of coherence in a system is usually impossible and must therefore be extracted from another system acting as a reservoir. In this article we present two faithful extraction protocols in the sense that the interaction involved between the system and reservoir is strictly coherence preserving. As an example we implement both in the case where the reservoir is a quantum harmonic oscillator in a coherent and a squeezed state respectively, and study the limits of repeatable extraction. For a single extraction onto qubits it is demonstrated that, perhaps surprisingly, one of the protocols manages to outperform a previous method, known as catalytic coherence, which allows the creation of an extra amount of coherence between degenerate energy eigenstates of the combined system. 
\end{abstract}
\maketitle
%%%%%%%%
\section{Introduction}
One of the most striking features of quantum mechanics is the notion of \emph{superposition}, i.e., the idea that a quantum system can exist in different states simultaneously, whether this is an electron passing through both slits of a screen during interference experiments or a cat which is both dead and alive inside a box. Notions like these have been rigorously defined both qualitatively and quantitatively \cite{aberg2006quantifying} and recast into the \emph{resource theory of coherence} \cite{PhysRevLett.113.140401,PhysRevLett.116.120404,RevModPhys.89.041003} and the \emph{resource theory of quantum reference frames and asymmetry} \cite{RevModPhys.79.555,Gour_2008,Marvian2014,PhysRevA.90.062110} with applications ranging from metrology \cite{PhysRevA.94.052324} and biology \cite{Lloyd_2011}, to thermodynamics \cite{PhysRevLett.113.150402,Lostaglio2015,PhysRevLett.115.210403,PhysRevX.5.021001,Narasimhachar2015,Korzekwa_2016} and the theory of entanglement \cite{PhysRevLett.115.020403,PhysRevLett.117.020402}.

The amount of coherence present in a system is a useful resource which enables one to lift restrictions imposed by conservation laws and simulate transformations which would otherwise be impossible. For example conservation of energy forbids the creation of a pure state in a superposition of different energy levels, from a system which starts initially in a state of definite energy. The only way to achieve this transformation is by extracting the desired superposition from another system, which acts as a reservoir, using a \emph{coherence extraction protocol}. 

In \cite{PhysRevLett.113.150402} such a protocol was proposed known as \emph{catalytic coherence} in which coherence can be extracted to a qubit initially prepared in the ground state of its Hamiltonian by interacting with a half-infinite ladder system in a superposition of it's energy eigenstates, through an energy-conserving unitary operation. By construction, this process is repeatable, at the cost of some fixed amount of energy each time, allowing one to extract in principle an arbitrarily large amount of coherence from the reservoir. 

In Sec. III it will be shown that the interaction involved in the above protocol creates an additional amount of coherence between degenerate eigenstates of the combined system's Hamiltonian. Furthermore this amount, is always greater than what is eventually extracted. An additional drawback lies in the fact that the protocol can only be applied to reservoirs with an infinite number of energy levels \cite{PhysRevLett.123.020403,PhysRevLett.123.020404}, (see also \cite{Vaccaro_2018} for additional criticisms regarding correlations between the extracted qubits). 

Motivated by this we focus attention on \emph{faithful} extraction protocols in which the interactions involved are strictly coherence preserving. This guarantees that the coherence gets extracted from the reservoir and is not introduced in some other way, making them suitable for studying degradation effects \cite{Bartlett_2006,Poulin_2007,doi:10.1080/09500340701289254,PhysRevLett.80.2023}. In Sec. IV two examples of such protocols, one able to extract a smaller and the other a larger amount of coherence each time, are given which can be used on reservoirs with a finite as well as an infinite number of energy levels. In Sec. V these are implemented for a reservoir, in a coherent and a squeezed state of the quantum harmonic oscillator respectively. After a short discussion on the limits of repeatability, it will be shown that for the second protocol coherence extraction to qubits is more efficient than what is possible with \cite{PhysRevLett.113.150402}. 

We begin by giving a short introduction to the resource theory of quantum coherence as well as a general description of coherence extraction protocols.
%%%%%%%%
\section{Resource theory of quantum coherence}
As in any \emph{resource theory} (see \cite{RevModPhys.89.041003,RevModPhys.91.025001} for a recent review) the resource theory of quantum coherence is defined by the set of \emph{free} or \emph{incoherent} states $\mathcal{I}$ and the set of \emph{free} or \emph{incoherent} operations $\mathcal{L}$. Let ${A}$ denote a Hermitian observable of interest. We will consider first the situation in which the spectrum of $A$ is non-degenerate. In this case the set of incoherent states is equal to all those density operators which commute with ${A}$ 
    \begin{equation}\label{eq1}
        \mathcal{I}({A}):=\left\{\rho\left|[\rho,{A}]=0\right.\right\}.
    \end{equation}

The set of incoherent operations is now defined as those \emph{completely positive and trace-preserving operations} (CPTP), $\Lambda$, mapping $\mathcal{I}({A})$ to itself
    \begin{equation}\label{eq2}
        \mathcal L({A}):=\left\{\Lambda\in (CPTP)\left|\Lambda(\rho)\in \mathcal{I}({A}),\forall \rho\in\mathcal{I}({A})\right.\right\}.
    \end{equation}
By demanding ${A}$ obey a conservation law, this set can further be restricted to all those operations $\Lambda\in{\mathcal{L}}({A})$ satisfying
    \begin{equation}\label{eq3}
       \tr(A\rho)=\tr( A\Lambda(\rho))\quad\forall\rho.
    \end{equation}
In the following the set of incoherent operations conserving $A$ will be denoted by $\bar{\mathcal{L}}(A)$.

By definition any state $\rho\not\in\mathcal{I}({A})$ is a resource. These states are called \emph{coherent} and their coherence can be quantified by a non-negative real function $C(\cdot)$ on the set of density matrices. Any true measure of quantum coherence must satisfy two important properties, 
\begin{enumerate}
    \item[i)] \emph{faithfulness}: i.e. $C(\rho)=0$ iff $\rho\in\mathcal{I}({A})$ and
    \item[ii)] \emph{monotonicity under incoherent operations}: i.e. $C(\Lambda(\rho))\leq C(\rho)$, $\forall \Lambda\in \mathcal{L}({{A}})$.
\end{enumerate} 
An example of such a measure is given by the \emph{$\ell_1$-norm of coherence} \cite{PhysRevLett.113.140401}
    \begin{equation}\label{eq4}
        C_{\ell_1}(\rho)=\sum_{i\neq j}\abs{\rho_{ij}}
    \end{equation}
where $\rho_{ij}$ are the non diagonal elements of $\rho$ in the eigenbasis of ${A}$.

In the case of a degenerate spectrum, the set of incoherent states and operations as well as those conserving $A$ is given again by  \crefrange{eq1}{eq3}. This time states with coherence between degenerate eigenstates of $A$ belong to $\mathcal{I}(A)$ and \cref{eq4} splits into two parts
\begin{equation}\label{eq5}
    C_{\ell_1}(\rho)=C_{usef}(\rho)+C_{free}(\rho),
\end{equation}
where
\begin{equation}\label{eq6}
    C_{free}(\rho)=C_{\ell_1}(\Pi(\rho))
\end{equation}
is the amount of degenerate coherence which can be created for free by the action of a quantum operation belonging to $\mathcal{L}(A)$ on any completely diagonal state and is stored in the free state
\begin{equation}\label{eq7}
   \Pi(\rho)=\sum_i P_i\rho P_i ,
\end{equation}
where $P_i$ is the projection onto the eigenstates of $A$ with the same eigenvalue $a_i$, and
\begin{equation}\label{eq8}
    C_{usef}(\rho)=C_{\ell_1}(\rho)-C_{\ell_1}(\Pi(\rho))
\end{equation}
is the amount of \emph{useful} coherence between non-degenerate eigenstates stored in the state due to a violation of \cref{eq1}.
%%%%%%%%%%%
\subsection{Coherence extraction protocols}    
Let ${A}_S$ and ${B}_R$ be two Hermitian operators. With the help of a reservoir $R$ in state $\sigma_R\not\in\mathcal{I}({B}_R)$ containing coherence with respect to observable $B_R$ and acting as a reservoir we can simulate a \emph{coherent channel} $\Phi_{\sigma_R}\not\in\mathcal{L}({A}_S)$ acting on Hilbert space $\mathscr{H}_S$.

Specifically suppose $\rho_S\in\mathcal{I}({A}_S)$ is initially incoherent. The desired channel is constructed by applying an incoherent operation $\Lambda\in\mathcal{L}({A}_S+{B}_R)$ on the composite system followed by tracing out $R$
    \begin{equation}\label{eq9}
        \Phi_{\sigma_R}(\rho_S)=\tr_R{\Lambda(\rho_S\otimes\sigma_R)}.
    \end{equation}
Similarly we can also define the induced quantum channel $\Psi_{\rho_S}\in\mathcal{L}({B}_R)$, acting on $\mathscr{H}_R$ by
     \begin{equation}\label{eq10}
        \Psi_{\rho_S}(\sigma_R)=\tr_S{\Lambda(\rho_S\otimes\sigma_R)}.
    \end{equation}
Since $\Phi_{\sigma_R}\not\in\mathcal{L}({A}_S)$ and $\Psi_{\rho_S}\in\mathcal{L}({B}_R)$, it follows that for any measure $C(\cdot)$
    \begin{equation}\label{eq11}
        C\left(\Phi_{\sigma_R}(\rho_S)\right)\geq 0
    \end{equation}
and
    \begin{equation}\label{eq12}
        C\left(\Psi_{\rho_S}(\sigma_R)\right)\leq C(\sigma_R).
    \end{equation}
As a result coherence has been extracted from the reservoir and stored in system $S$. The protocol associated with \cref{eq9,eq10} is called a \emph{coherence extraction protocol}.

The maximum possible amount of extractable coherence is known as the \emph{cohering power} of the channel and is given by \cite{PhysRevA.92.032331,BU20171670}
    \begin{equation}\label{eq13}
        \mathcal{C}(\Phi)=\max_{\rho_S\in\mathcal{I}({A}_S)}C(\Phi(\rho_S)).
    \end{equation}
Equation (\ref{eq13}) provides a measure of the efficiency of the protocol.
%%%%%%%%
\section{Catalytic extraction protocol}
The observables of interest in this case are the Hamiltonians  ${H}_S=\epsilon_0\ketbra{1}$ and ${H}_R=\epsilon_0\sum_{n=0}^{\infty}n\ketbra{n}$ of a qubit and the reservoir. Note that the reservoir is a system with a fixed energy difference between consecutive levels, equal to $\epsilon_0$, which matches the excited energy of the qubit.

The protocol consists of two stages \cite{PhysRevLett.113.150402}. The first stage shifts the reservoir up one level
\begin{equation}\label{eq14}
    \mathfrak{D}(\sigma_R)=\Delta\sigma_R\Delta^{\dagger}
\end{equation}
where $\Delta=\sum_{n=0}^{\infty}\ketbra{n+1}{n}$ is the \emph{shift operator}. Since $\Delta^{\dagger}\Delta=I$, it follows that $\mathfrak{D}$ is a trace preserving quantum operation. On the other hand, this step requires an amount of energy equal to $\epsilon_0$ to be consumed in the process, so $\mathfrak{D}$ is not energy conserving. This does not affect the discussion however since, as we shall see in Sec. IV, it can always be extended to an energy conserving unitary interaction between $\sigma_R$ and an additional qubit in it's excited state. This additional qubit is no longer needed for the rest of the protocol so we can safely ignore it's existence.

The second stage of the protocol consists of the following energy conserving unitary interaction between the two systems
\begin{equation}\label{eq15}
    V_+(U)=\left(\begin{array}{cc}
         \ketbra{0}+U_{00}\Delta\Delta^\dagger&U_{01}\Delta  \\
         U_{10}\Delta^\dagger&U_{11}I 
    \end{array}\right)
\end{equation}
where each block acts on $\mathscr{H}_R$ and $U_{ij}$ are the elements of some unitary operator $U$ acting on $\mathscr{H}_S$. 

Suppose that initially $\rho_S=\ketbra{0}$. From \cref{eq9,eq10} we find that
    \begin{equation}\label{eq16}
        \Phi_{\sigma_R}(\ket{0})=
        \left(\begin{array}{cc}
             \abs{U_{00}}^2 &U_{00}U^*_{10}\tr(\Delta\sigma_R)  \\
 U^*_{00}U_{10}\tr(\Delta^{\dagger}\sigma_R)& \abs{U_{10}}^2
        \end{array}\right)
    \end{equation}
and
    \begin{equation}\label{eq17}
        \Psi_{\ket{0}}(\sigma_R)=\abs{U_{00}}^2\Delta\sigma_R\Delta^{\dagger}+\abs{U_{10}}^2\sigma_R.
    \end{equation}
A key element of the protocol lies in the fact that $\tr(\Delta\Psi_{\ket{0}}(\sigma_R))=\tr(\Delta\sigma_R)$. This means that the process can be repeated with $\Psi_{\ket{0}}(\sigma_R)$ acting as the new reservoir. Provided a sufficient amount of energy, we can retrieve a sequence of qubits all in the same state $\Phi_{\sigma_R}(\ket{0})$. It thus appears that it is possible to extract an arbitrarily large amount of coherence from the reservoir which acts as some kind of catalyst. This phenomenon 
is also known as the \emph{coherence embezzling phenomenon} \cite{2019arXiv190609067C}.

Let us now compute the amount of extracted coherence stored in the qubit. Using the $\ell_1$-norm as a measure we find from \cref{eq16}
    \begin{equation}\label{eq18}
    {C}_{\ell_1}(\Phi_{\sigma_R}(\ket{0}))=2\abs{U_{00}}\abs{U_{10}}\abs{\tr(\Delta\sigma_R)}.
    \end{equation}
On the other hand for any $n\geq 1$
\begin{equation}\label{eq19}
    V_+(U)(\ket{0}\otimes\ket{n})=U_{00}\ket{0}\otimes\ket{n}+U_{10}\ket{1}\otimes\ket{n-1}.
\end{equation}
%\begin{equation}\label{eq19}
%        \mathcal{C}_{free}(V_+(U))=2\max_{j=0,1}\abs{U_{0,j}}\abs{U_{1,j}}
%    \end{equation}
With the help of \cref{eq6} we see that the interaction is actually responsible for creating an amount of
\begin{equation}
    2\abs{U_{00}}\abs{U_{10}}
\end{equation}
units of free coherence between degenerate eigenstates of the combined system. Equation (\ref{eq18}) depends on this extra amount and $\abs{\tr(\Delta\sigma_R)}$ which is  a measure of the coherence originally present in the reservoir. For this reason catalytic coherence cannot be considered as a true extraction protocol and is not suited for studying degradation effects in the reservoir \cite{Bartlett_2006,Poulin_2007,doi:10.1080/09500340701289254,PhysRevLett.80.2023}.

Moreover since $\abs{\tr(\Delta\sigma_R)}\leq1$, it immediately follows that
    \begin{equation}\label{eq20}
        {C}_{\ell_1}(\Phi_{\sigma_R}(\ket{0}))\leq  2\abs{U_{00}}\abs{U_{10}}
    \end{equation}
and therefore the amount of extracted coherence is always less than that created in the combined system by the interaction. 
%is the cohering power of the interaction. Equation (\ref{eq18}) depends on two independent factors. The first factor, $\mathcal{C}_{\ell_1}(V_+(U))$, represents the maximum amount of coherence that the interaction can create between degenerate eigenstates of the combined system, while the second, $\abs{\tr(\Delta\sigma_R)}$, is  a measure of the amount of coherence that is originally present in $\sigma_R$. Since $\abs{\tr(\Delta\sigma_R)}\leq1$, it immediately follows that
%    \begin{equation}\label{eq20}
%        \mathcal{C}_{\ell_1}(\Phi_R)\leq \mathcal{C}_{\ell_1}(V_+(U))
%    \end{equation}
%and therefore the maximum amount of extracted coherence is always less than the amount created in the combined system by the interaction.
%%%%%%%%
\section{Faithful extraction protocols}
We will now present two extraction protocols which are faithful, in the sense that the interactions involved between the two systems have zero cohering power and are thus incapable of creating any extra amounts of coherence on the combined system. For simplicity we will always assume that the systems to which coherence is stored are qubits. The two protocols, a weak version which is able to extract only a small amount of coherence each time and the other, a stronger version, able to extract a larger amount, are distinguished by the fact that in the former only a single qubit is needed each time, while in the latter, the number of qubits necessary increases exponentially with respect to the number of extractions. It is worth mentioning that unlike the case of catalytic extraction, both protocols can be applied to reservoirs with a finite as well as an infinite number of energy levels.
\subsection{Weak faithful extraction}
\begin{figure}
    \includegraphics[width=\columnwidth]{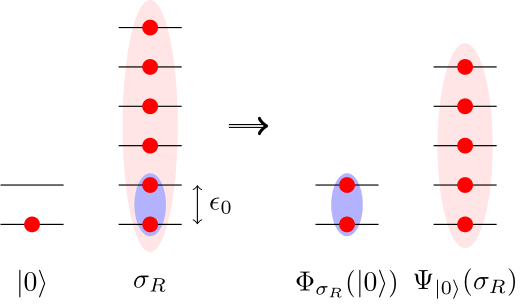}
    \caption{Weak faithful extraction of quantum coherence from a 6-level energy ladder to a qubit initially in the ground state.  The protocol extracts the coherence between the ground and first excited state of the reservoir and stores this amount into the qubit. After the interaction the reservoir has lost a quantum of energy. The process can be repeated at most 5 times.}
    \label{fig1}
\end{figure}
Let's consider the general case in which we wish to extract coherence from a finite energy ladder reservoir with $N+1$ levels and Hamiltonian $ H_R=\epsilon_0\sum_{n=0}^{N}n\ketbra{n}$ to a qubit with the same Hamiltonian as before. The interaction between the reservoir and the qubit in this case is given by
\begin{equation}\label{eq21}
    V=\left(\begin{array}{cc}
         \ketbra{0}& \Delta \\
        \Delta^\dagger & \ketbra{N}
    \end{array}\right)
\end{equation}
where $\Delta=\sum_{n=0}^{N-1}\ketbra{n+1}{n}$. From $\Delta\Delta^{\dagger}=I-\ketbra{0}$ and $\Delta^{\dagger}\Delta=I-\ketbra{N}$, it can be checked that $V$ is unitary and also conserves the total energy. Moreover since for any $0\leq n\leq N$, both
\begin{equation}\label{eq22}
    V(\ket{0}\otimes\ket{n})=\delta_{0n}\ket{0}\otimes\ket{0}+(1-\delta_{0n})\ket{1}\otimes\ket{n-1}
\end{equation}
 and 
 \begin{equation}\label{eq23}
    V(\ket{1}\otimes\ket{n})=(1-\delta_{Nn})\ket{0}\otimes\ket{n+1}+\delta_{Nn}\ket{1}\otimes\ket{N}
 \end{equation} 
 are incoherent it follows that $V$ is coherence conserving, ($\mathcal{C}(V)=0$), so $V\in\bar{\mathcal{L}}({H}_S+{H}_R)$.

If initially $\rho_S=\ketbra{0}$ then after the interaction
\begin{equation}\label{eq24}
    \Phi_{\sigma_R}(\ket{0})=\left(\begin{array}{cc}
         \sigma_{00}&\sigma_{01}\\
         \sigma^*_{01}&1-\sigma_{00}
    \end{array}\right),
\end{equation}
where $\sigma_{nn'}=\bra{n}\sigma_R\ket{n'}$, are the matrix elements of $\sigma_R$ and
\begin{equation}\label{eq25}
    \Psi_{\ket{0}}(\sigma_R)=\sigma_{00}\ketbra{0}+\Delta^{\dagger}\sigma_R\Delta.
\end{equation}
From \cref{eq24} it can be seen that the amount of coherence extracted in this case, measured using the $\ell_1$-norm, is equal to that between the ground and excited states of the reservoir which loses a quantum of energy in the process (see \cref{fig1}).

Repeating the procedure we find by induction that the amount of extracted coherence after $m$ applications of the protocol is equal to
%\begin{equation}\label{eq23}
%    \rho_S^{(m)}=\left(
%    \begin{array}{cc}
%         \sum_{k=0}^{m-1}\sigma_{kk}&\sigma_{m-1,m}  \\
%        \sigma^*_{m-1,m} & 1-\sum_{k=0}^{m-1}\sigma_{kk}
%    \end{array}\right),
%\end{equation}
\begin{equation}\label{eq26}
    C_{\ell_1}(\rho_S^{(m)})=2\abs{\sigma_{m-1,m}}
\end{equation}
units of coherence, while the state of the reservoir after each extraction is given by
\begin{equation}\label{eq27}
    \sigma_R^{(m)}=\sum_{k=0}^{m-1}\sigma_{kk}\ketbra{0}+(\Delta^{\dagger})^m\sigma_R\Delta^m.
\end{equation}
Since $\Delta^{N+1}=0$, it follows that after $N$ repetitions of the protocol the reservoir is in it's ground state and extraction is no longer possible. 

For a reservoir with an infinite number of energy levels ($N\to\infty$), the interaction is given by \cref{eq15} with $U$ the $\sigma_x$ Pauli matrix
\begin{equation}\label{eq28}
    V_+(\sigma_x)=\left(\begin{array}{cc}
         \ketbra{0}&\Delta \\
         \Delta^\dagger&0 
    \end{array}\right).
\end{equation}
This time extraction is only possible for a qubit initially in it's ground state for which \cref{eq26,eq27} remain the same. If the qubit is excited, then after the interaction it will relax to it's ground state and the new state of the reservoir will be given by \cref{eq14}. Note that $V_+(\sigma_x)$ is simply the unitary interaction that is necessary in order to implement the first step of the catalytic extraction protocol. 
\subsection{Strong faithful extraction}
\begin{figure}
    \includegraphics[width=\columnwidth]{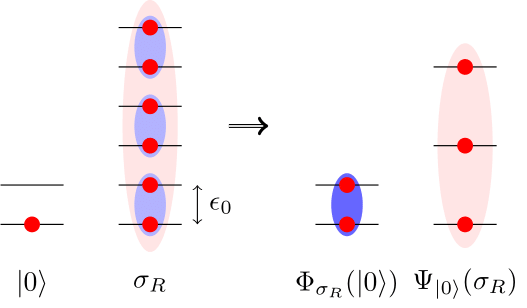}
\caption{Strong faithful extraction of quantum coherence from a 6-level energy ladder to a qubit initially in the ground state. The protocol extracts the coherence between disjoint pairs of consecutive energy levels of the reservoir and stores this amount into the qubit. After extraction any remaining coherence between energy levels with distance $\epsilon_0$ in the reservoir has been destroyed. The process can be repeated to extract coherence from energy levels with distance $2\epsilon_0$ onto a pair of qubits. The protocol cannot be repeated more than 3 times.}
\label{fig2}
\end{figure}
Consider now the following interaction
\begin{equation}\label{eq29}
    V=\left(\begin{array}{cc}
        P_2 & \Delta P_2 \\
        P_2\Delta^\dagger &\Delta P_2\Delta^\dagger 
    \end{array}\right),
\end{equation}
where 
\begin{equation}\label{eq30}
    P_2=\sum_{n=0}^{N/2-1}\ketbra{2n}
\end{equation} 
is the projection onto the subspace spanned by the even eigenstates of the reservoir's Hamiltonian, where we have also tacitly assumed that the number of energy levels is also even. This leads to no loss of generality, since a reservoir with an odd number of energy levels can always be thought of as being part of some larger system. Making use of the fact that $\Delta P_2\Delta^{\dagger}=I-P_2$ and $P_2\Delta P_2=0$ it can be shown that once again $V$ is unitary and energy conserving, $[ H_{tot},V]=0$ where $H_{tot}=H_S+H_R$ is the total Hamiltonian of the combined system. 

Since for $i,j=0,1$ and any $0\leq n\leq N/2-1$
\begin{equation}\label{eq31}
    V(\ket{i}\otimes\ket{2n+j})=\ket{j}\otimes\ket{2n+i},
\end{equation}
it follows that $\mathcal{C}(V)=0$ and $V\in\bar{\mathcal{L}}({H}_S+{H}_R)$, so the protocol is faithful.

Suppose that $\rho_S=\ketbra{0}$, then after the interaction
    \begin{equation}\label{eq32}
        \Phi_{\sigma_R}(\ket{0})=\left(\begin{array}{cc}
             \tr(P_2\sigma_R)&\tr(\Delta P_2\sigma_R)  \\
             \tr^*(\Delta P_2\sigma_R)&1-\tr(P_2\sigma_R)
        \end{array}\right),
    \end{equation}
\begin{equation}\label{eq33}
    \Psi_{\ket{0}}(\sigma_R)=P_2\sigma_RP_2+P_2\Delta^{\dagger}\sigma_R\Delta P_2,
\end{equation}
and the amount of extracted coherence is equal to
\begin{equation}\label{eq34}
    C_{\ell_1}(\Phi_{\sigma_R(\ket{0})})=2\abs{\tr(\Delta P_2\sigma_R)}.
\end{equation}
Expanding $\tr(\Delta P_2\sigma_R)=\sum_n{\sigma_{2n,2n+1}}$, it can be seen that the protocol essentially extracts the coherence between disjoint pairs of consecutive energy levels of the reservoir and stores this amount into the qubit (see \cref{fig2}). Comparing this case with that discussed previously, it is expected that for a single extraction from the same reservoir, this protocol will generally outperform the weaker one (for strong faithful extraction to systems with more energy levels see Supplementary). 

From \cref{eq33} we observe that because the reservoir has been projected onto the subspace of even energy levels, any remaining coherence between levels with energy difference equal to $\epsilon_0$ has now been destroyed. In order to extract coherence a second time we now need a pair of qubits both in their ground state. Treating this pair as an effective two level system with excited energy $2\epsilon_0$ we can substitute $P_2\to P_4=\sum_n\ketbra{4n}$, $\Delta\to\Delta^2P_2=\sum_n\ketbra{2n+2}{2n}$ and $\sigma_R\to\Psi_{\ket{0}}(\sigma_R)$ in \cref{eq33,eq34} to calculate the newly extracted amount. Iterating this process it can be shown by induction that after $m$ extractions an amount of 
\begin{equation}\label{eq35}
    C_{\ell_1}(\rho_S^{(m)})=2\abs{\tr(\Delta^{2^{m-1}}P_2^{(m)}\sigma_R)}.
\end{equation}
units of coherence has been stored onto a system of $2^{m-1}$ qubits with combined Hamiltonian
\begin{equation}\label{eq36}
    {H}_S^{(m)}=\bigoplus_{i=1}^{2^{m-1}}H_S
\end{equation}
where 
\begin{equation}\label{eq37}
    P_2^{(m)}=\sum_{k=0}^{2^{m-1}-1}\Delta^kP_{2^m}(\Delta^{\dagger})^k
\end{equation}
and $P_{2^m}=\sum_{n}\ketbra{2^mn}$ where the summation is taken over those integer values $n\leq N/2^m-1$. 

In a similar fashion the state of the reservoir after each extraction will be equal to
\begin{equation}\label{eq39}
    \sigma_R^{(m)}=P_{2^m}\left(\sum_{k=0}^{2^m-1}(\Delta^\dagger)^k\sigma_R\Delta^k\right)P_{2^m}.
\end{equation}
Since the total energy of all extracted qubits cannot exceed that of the highest occupied energy of the reservoir ($(N-1)\epsilon_0$), the protocol can be repeated at most $\lfloor\log_2N\rfloor$ times.
%%%%%%%%%
\section{Extraction from a quantum harmonic oscillator}
We will now implement both protocols in the case where the reservoir is a quantum harmonic oscillator in the coherent state \cite{PhysRevLett.10.84}
    \begin{equation}\label{eq40}
        \ket{a}=e^{-\frac{\abs{a}^2}{2}}\sum_{n=0}^{\infty}\frac{(e^{i\phi}\abs{a})^n}{\sqrt{n!}}\ket{n},
    \end{equation}
as well as the single mode squeezed vacuum state \cite{SCHNABEL20171}
\begin{equation}\label{eq41}
    \ket{SMSV}=\frac{1}{\sqrt{\cosh{r}}}\sum_{n=0}^\infty(-e^{i\phi}\tanh{r})^n\frac{\sqrt{(2n)!}}{2^{n}n!}\ket{2n}
    \end{equation}
where $\abs{a}$, $r\geq0$ are coherence parameters and $\phi$ a phase.
%%%%%%%%%%
\subsection{Weak faithful extraction}
\begin{figure}
\subfloat[Coherent reservoir]{\includegraphics[width=\columnwidth,trim=0 0 0 0.9cm,clip]{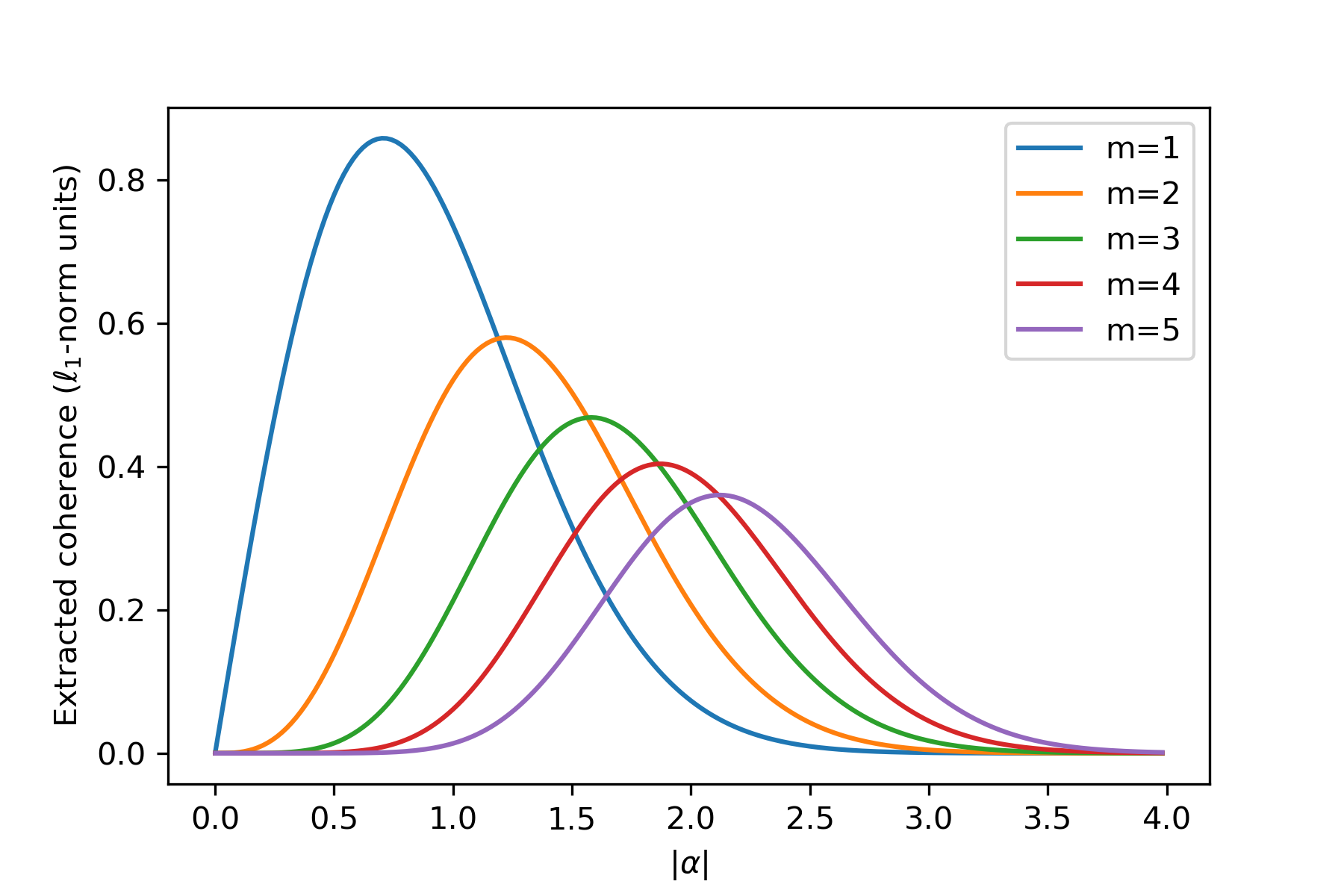}}\\
\subfloat[Squeezed reservoir]{\includegraphics[width=\columnwidth,trim=0 0 0 0.9cm,clip]{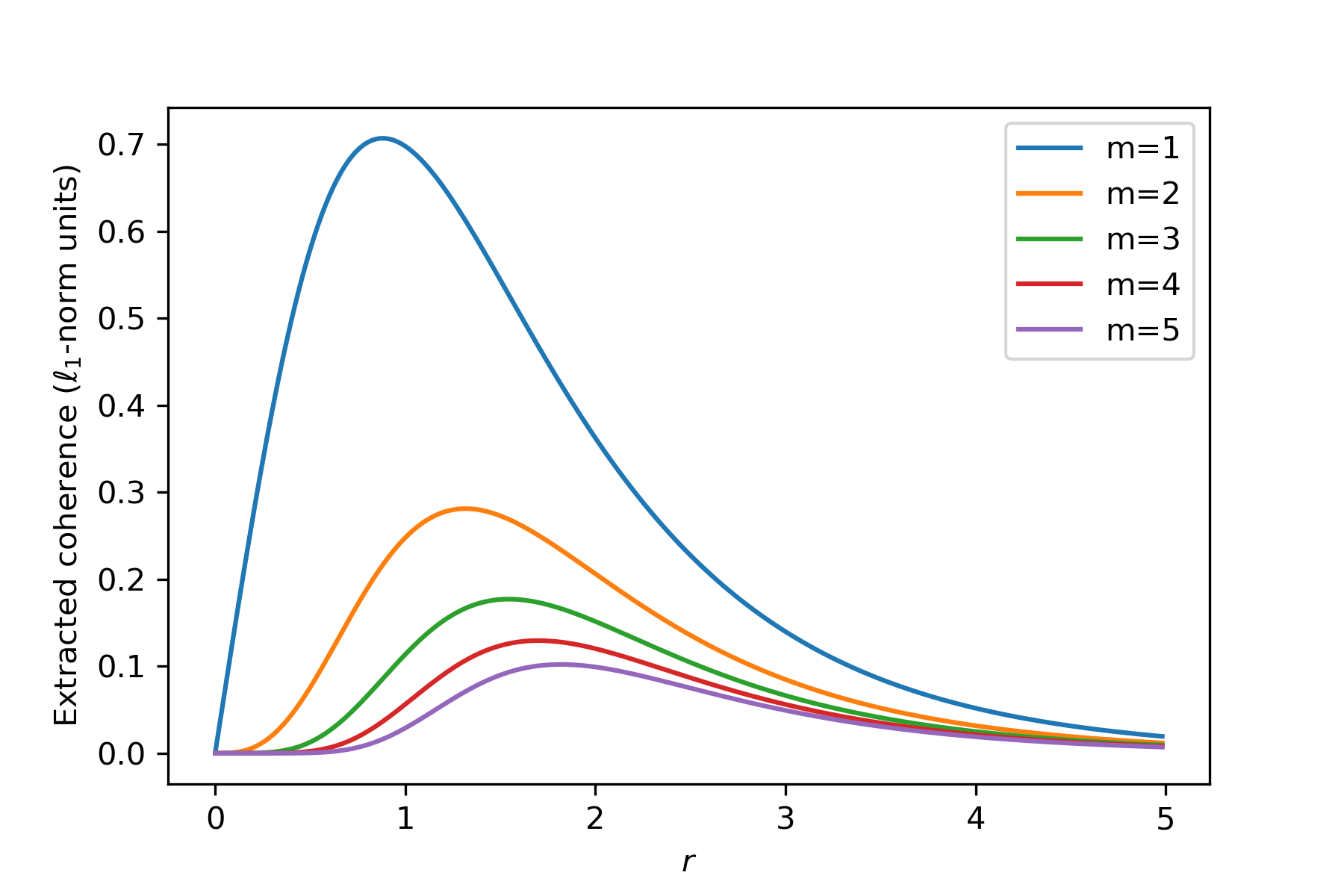}}
\caption{Amount of extracted coherence measured in $\ell_1-norm$ units from a quantum harmonic oscillator after (from top to bottom) $m=1,2,3,4,5$ applications of the weak extraction protocol.}\label{fig3}
\end{figure}
With the help of \cref{eq26} we find that the amount of extracted coherence after $m$ applications of the protocol is equal to
%\begin{equation}\label{eq37}
%    \rho_S^{(m)}=e^{-\abs{a}^2}\left(\begin{array}{cc}
%         \sum_{k=0}^{m-1}\frac{\abs{a}^{2k}}{k!}&e^{-i\phi}\frac{\abs{a}^{2m-1}}{\sqrt{(m-1)!(m)!}}  \\
%         e^{i\phi}\frac{\abs{a}^{2m-1}}{\sqrt{(m-1)!(m)!}}&1- \sum_{k=0}^{m-1}\frac{\abs{a}^{2k}}{k!}
%    \end{array}\right)
%\end{equation}
\begin{figure*}
\begin{minipage}{\columnwidth}
\subfloat{\includegraphics[width=\columnwidth,trim=0 0 0 0.9cm,clip]{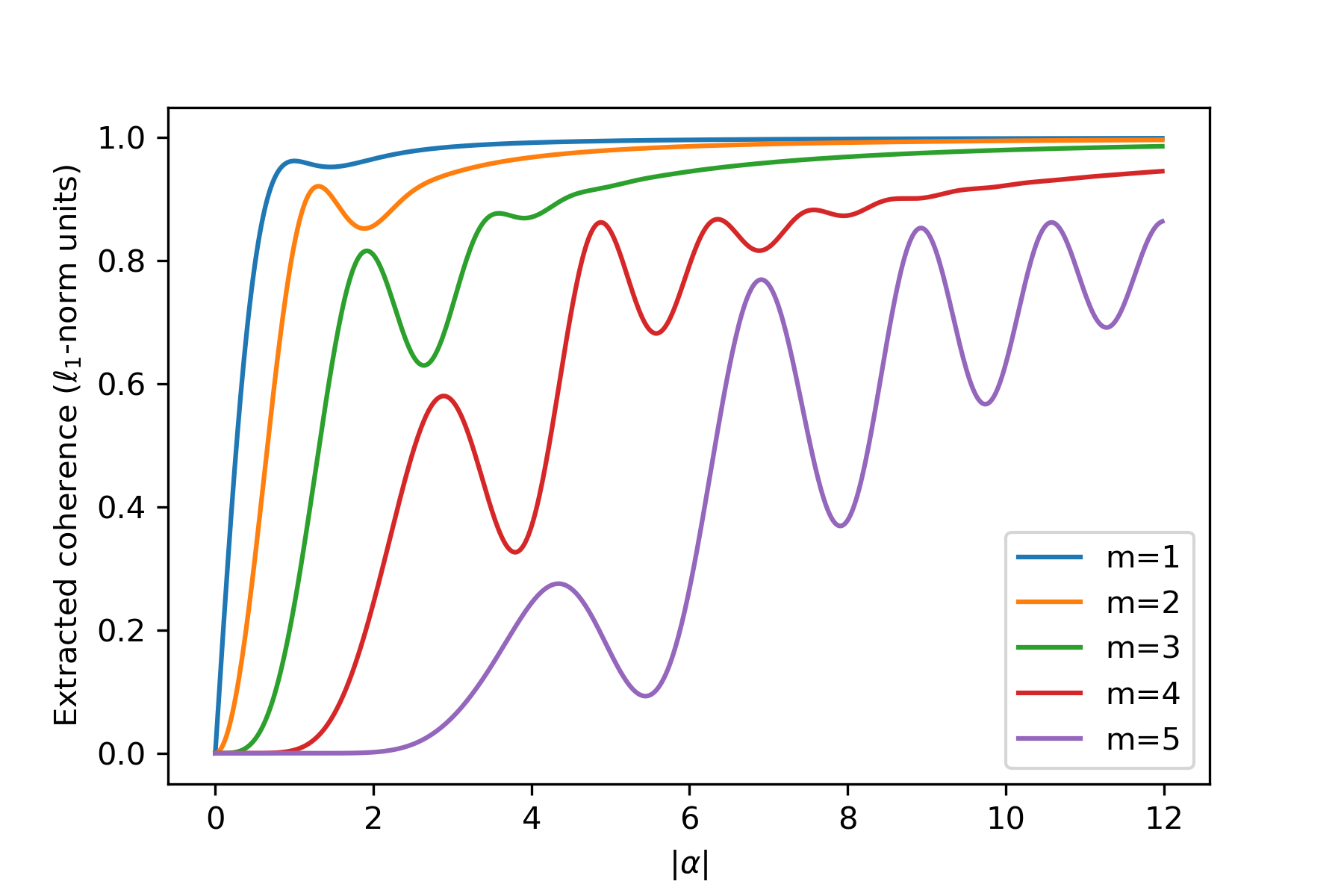}}\hfill\\
\subfloat{\includegraphics[width=\columnwidth,trim=0 0 0 0.9cm,clip]{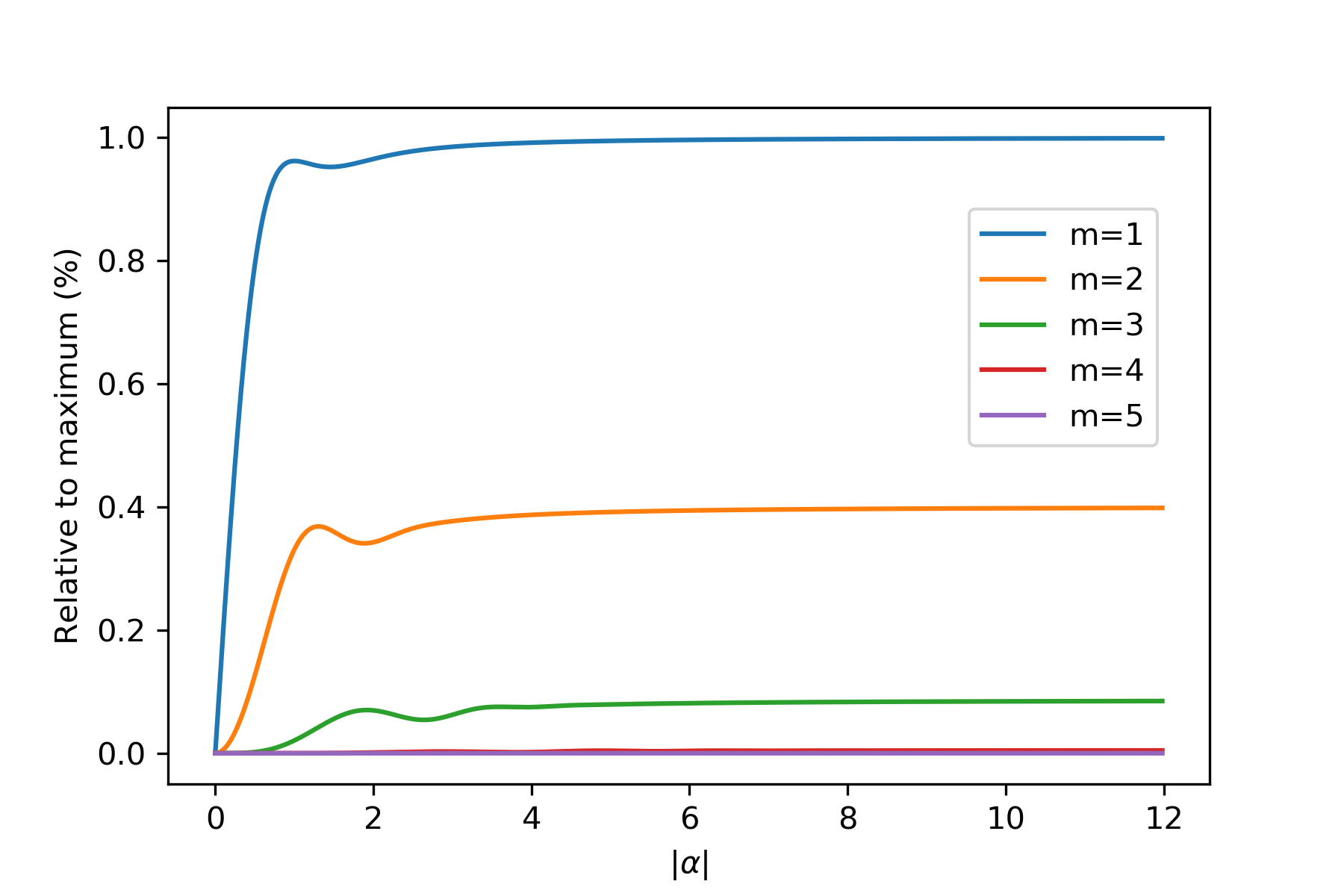}}\\
(a) Coherent reservoir
\end{minipage}
\begin{minipage}{\columnwidth}
\subfloat{\includegraphics[width=\columnwidth,trim=0 0 0 0.9cm,clip]{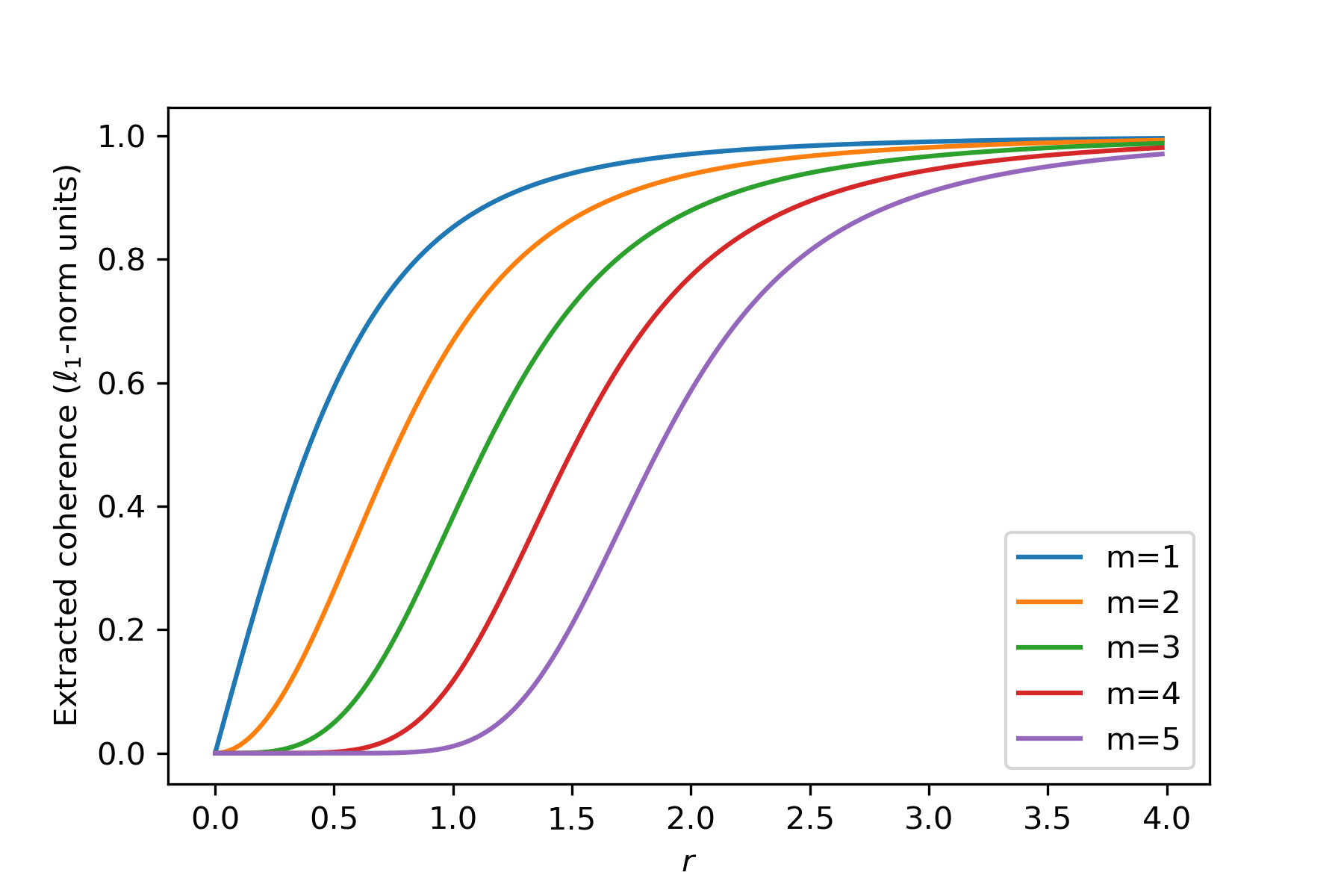}}\\
\subfloat{\includegraphics[width=\columnwidth,trim=0 0 0 0.9cm,clip]{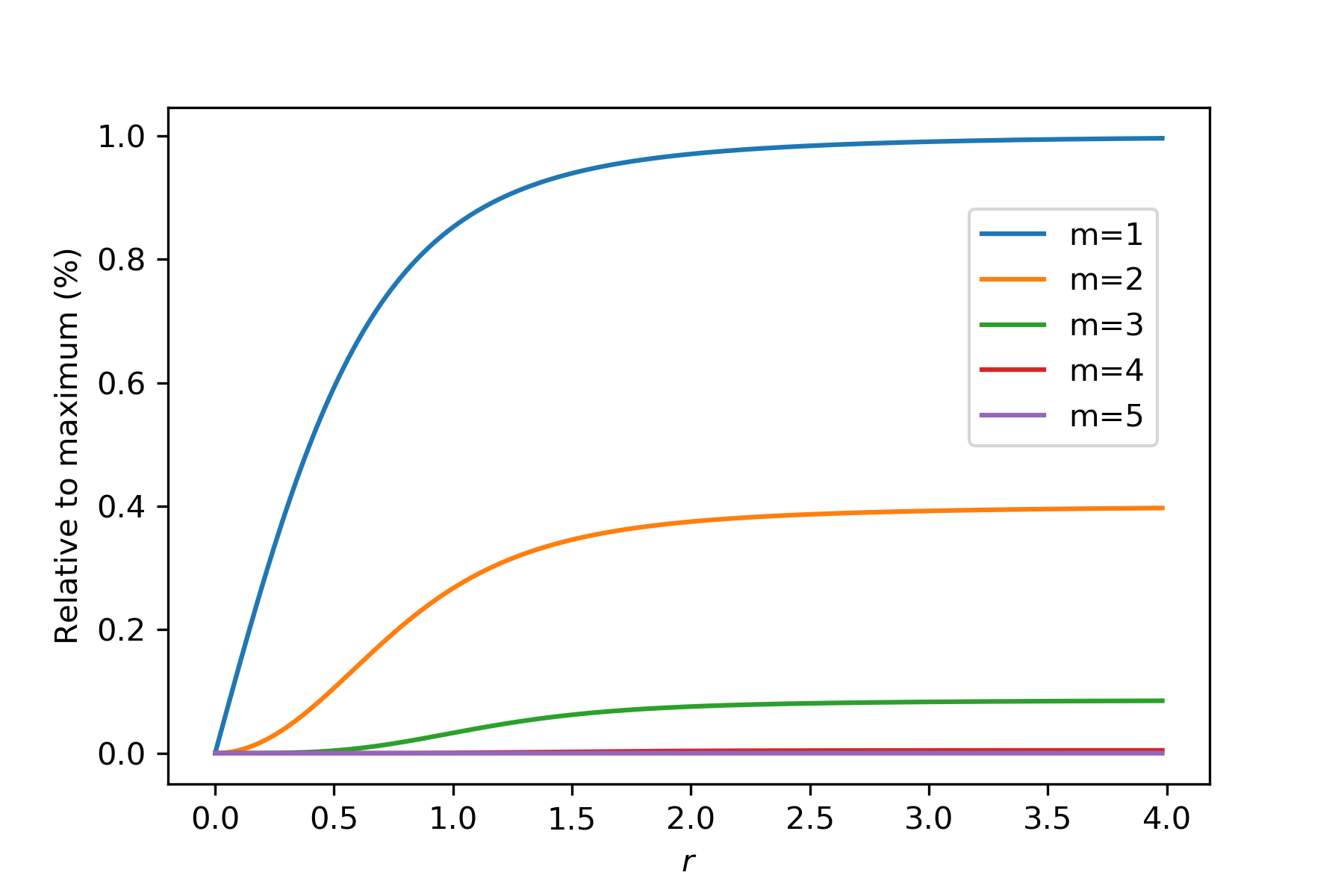}}\\
(b) Squeezed reservoir
\end{minipage}
\caption{Amount of extracted coherence from a harmonic oscillator after (from top to bottom) $m=1,2,3,4,5$ applications of the strong extraction protocol. (Upper) extracted amount in $\ell_1-norm$ units. (Bottom) relative to maximum useful amount for a system of $2^{m-1}$ qubits.}
\label{fig4}
\end{figure*}
\begin{equation}\label{eq42}
    C_{\ell_1}(\rho_S^{(m)})=2e^{-\abs{a}^2}\frac{\abs{a}^{2m-1}}{\sqrt{(m-1)!(m)!}},
\end{equation}
for a reservoir in the coherent state and
\begin{equation}\label{eq43}
    C_{\ell_1}(\rho_S^{(m)})=2\frac{(\tanh{r})^{2m-1}}{\cosh{r}}\frac{\sqrt{(2m-2)!(2m)!}}{2^{2m-1}(m-1)!m!}
\end{equation}
for the squeezed state. In both cases the amount of extracted coherence is independent of the phase. 

In \cref{fig3} we present the extracted amount as a function of the coherence parameter for different values of $m$. We observe that for the coherent reservoir this amount fluctuates depending the value of $\abs{a}$ while in the case of the squeezed reservoir the amount of extracted coherence decreases with each extraction. In both cases only a finite amount can be extracted in total since the maximum possible value decreases with $m$ as can be seen directly from \cref{eq42,eq43}.
%%%%%%%%%%%
\subsection{Strong faithfull extraction}
This time with the help of \cref{eq35} it can be shown that an amount of
\begin{equation}\label{eq44}
    C_{\ell_1}(\rho_S^{(m)})=2\sum_{k=0}^{2^{m-1}-1}F_{2^m;k,k+2^{m-1}}(\abs{a}^2)
\end{equation}
units of coherence gets extracted from the coherent reservoir, while for the squeezed case this amount is equal to
\begin{equation}\label{eq45}
    C_{\ell_1}(\rho_S^{(m)})=\frac{2}{\cosh{r}}\sum_{k=0}^{2^{m-1}-1}G_{2^m;k,k+2^{m-1}}\left((\tanh{r})^2\right)
\end{equation}
where the functions $F$ and $G$ are given by
    \begin{align}
        &F_{d;k,k'}(x)=e^{-x}\sum_{n=0}^\infty \frac{x^{nd+\frac{k+k'}{2}}}{\sqrt{(nd+k)!(nd+k')!}}\label{eq46}\\
        &G_{d;k,k'}(x)=\sum_{n=0}^{\infty}\left(\frac{x}{4}\right)^{nd+\frac{k+k'}{2}}\frac{\sqrt{(2nd+2k)!(2nd+2k')!}}{(nd+k)!(nd+k')!}.\label{eq47}
    \end{align}
Once again as in the weak case the amount of extracted coherence is independent of the phase.

In \cref{fig4} we present the extracted amount as a function of the coherence parameter for different values of $m$. For both states of the reservoir it seems that in the limit of very large parameter values the same amount of coherence gets extracted irrespective of the number of repetitions. This follows from the fact that for $\abs{a}\to\infty$ and $r\to 1$
\begin{equation}\label{eq48}
    \lim_{x\to\infty} F_{d,k,k'}(x)=\lim_{x\to 1}G_{d,k,k'}(x)\sqrt{1-x}=\frac{1}{d}.
\end{equation}
(for a formal proof see Supplementary). 

Because the strong protocol treats the combined system of $2^{m-1}$ qubits needed each time for extraction as an effective two level system, all of the coherence extracted gets stored between the ground and highest energy level of the system. This excludes a very large number of levels that could potentially be used for storing. The amounts given in \cref{eq44,eq45} should therefore be compared to the maximum possible amount of useful coherence which can be stored in the system. This is given by \cref{eq8} for $\rho$ equal to the maximaly coherent pure state of $d$ dimensions
\begin{equation}\label{eq49}
    \ket{\psi_d}=\frac{1}{\sqrt{d}}\sum_{i=0}^{d-1}\ket{i}
\end{equation}
which for a system of $M$ qubits is equal to
\begin{equation}\label{eq50}
    C_{max}=2^M-\frac{(2M)!}{2^M(M!)^2}.
\end{equation}
From \cref{fig4} it can be seen that compared to this amount extraction becomes negligible for both reservoirs after four repetitions.
%%%%%%%%
\section{Discussion}
\begin{table}[]
    \begin{tabular}{c||ccc}
        &Finite&Energy&Qubits\\
         &reservoir&(per repetition)&(per repetition)\\\hline
         Catalytic&No&$\epsilon_0$&1\\
         Weak&Yes&0&1\\
         Strong&Yes&0&$2^{m-1}$
    \end{tabular}
\caption{Comparison between catalytic coherence and faithful extraction protocols.}
\end{table}
Even though the catalytic coherence protocol is useful for implementing any coherent channel on a qubit (by choosing a reservoir with $\tr(\Delta\sigma_R)=1$ it can be shown that $\Phi_{\sigma_R}(\rho)=U\rho U^\dagger$), it cannot be considered as a true extraction protocol since there is always a free amount of coherence that gets injected into the combined system by the interaction, part of which is stored into the extracted system. As a matter of fact it was shown in \cref{eq20} that more coherence is actually injected than what is finally extracted. An additional drawback lies in the fact that it can only be applied to reservoirs with an infinite number of energy levels and also requires expenditure of an amount of energy equal to $\epsilon_0$ each time.

In contrast the two protocols developed in Sec. IV, which require no consumption of energy and can also be applied to any reservoir, are faithful since by construction the interactions involved are incapable of creating additional amounts of coherence. This in turn implies that any amount extracted must have necessarily originated in the reservoir. This fact is evident in \cref{fig3,fig4} where the amount of extracted coherence generally decreases for finite values of the coherence parameter due to degradation effects in the reservoir.

Although the stronger protocol in general is able to extract a larger amount than the weaker version, it was demonstrated that since the number of qubits required each time grows exponentially with each extraction, this amount as compared to the maximum amount of useful coherence that can in principle be stored in the system becomes very small after a few repetitions. It is nonetheless interesting to compare it's cohering power with that of catalytic coherence in the case of a single extraction from the coherent and squeezed reservoirs. Since for the strong protocol the amount of extracted coherence is independent of the initial state of the qubit, it's cohering power is equal to \cref{eq34}. In \cref{fig5} we compare this with the cohering power for the catalytic case which is given by $\abs{tr(\Delta\sigma_R)}$. It can be seen that for both reservoirs the strong protocol actually outperforms catalysis for any value of the coherence parameter. This is perhaps surprising considering the fact that in the former case no extra coherence has been introduced during extraction.
\begin{figure}
\subfloat[Coherent reservoir]{\includegraphics[width=\columnwidth,trim={0 0 0 0.8cm},clip]{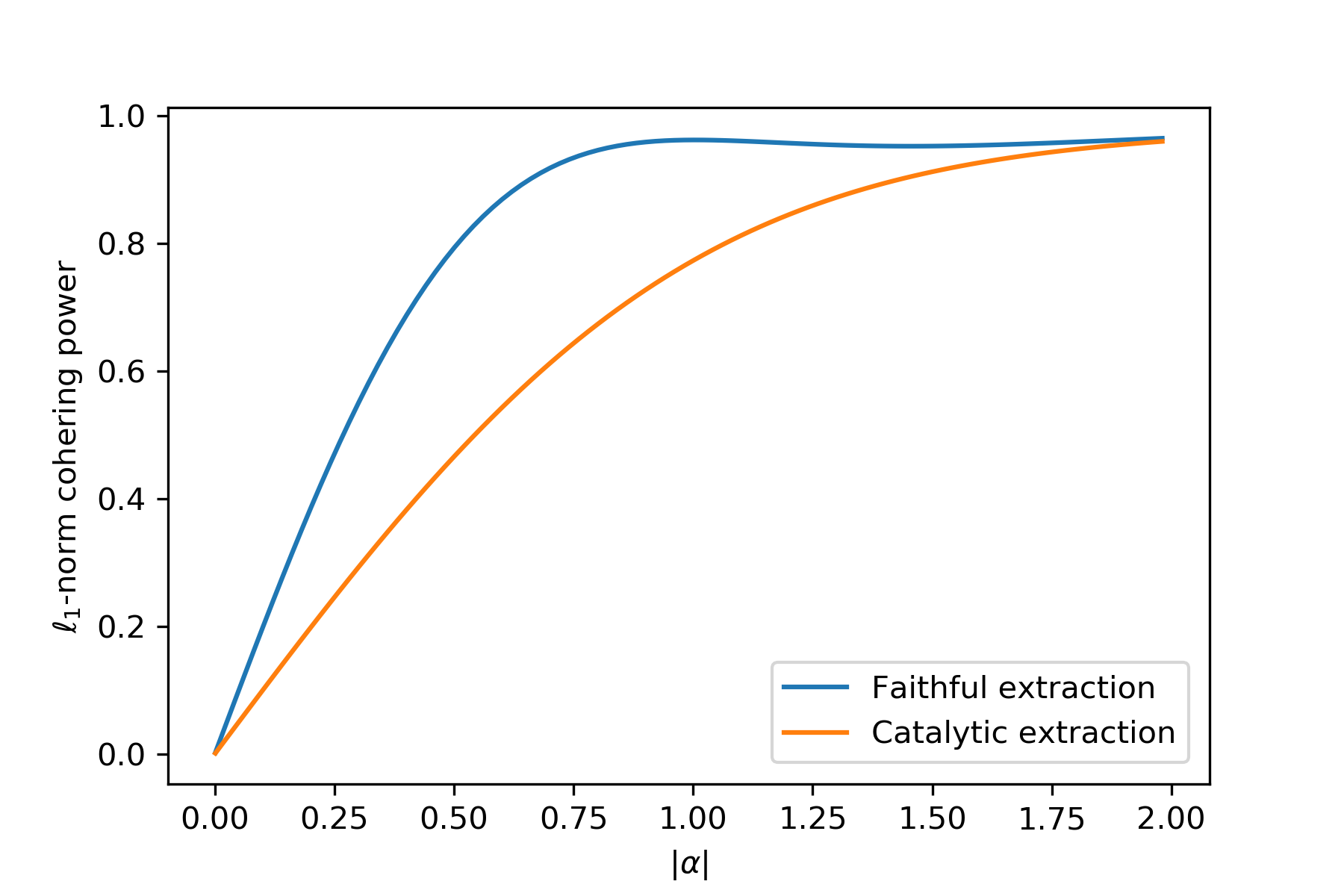}}\\
\subfloat[Squeezed reservoir]{\includegraphics[width=\columnwidth,trim={0 0 0 0.8cm},clip]{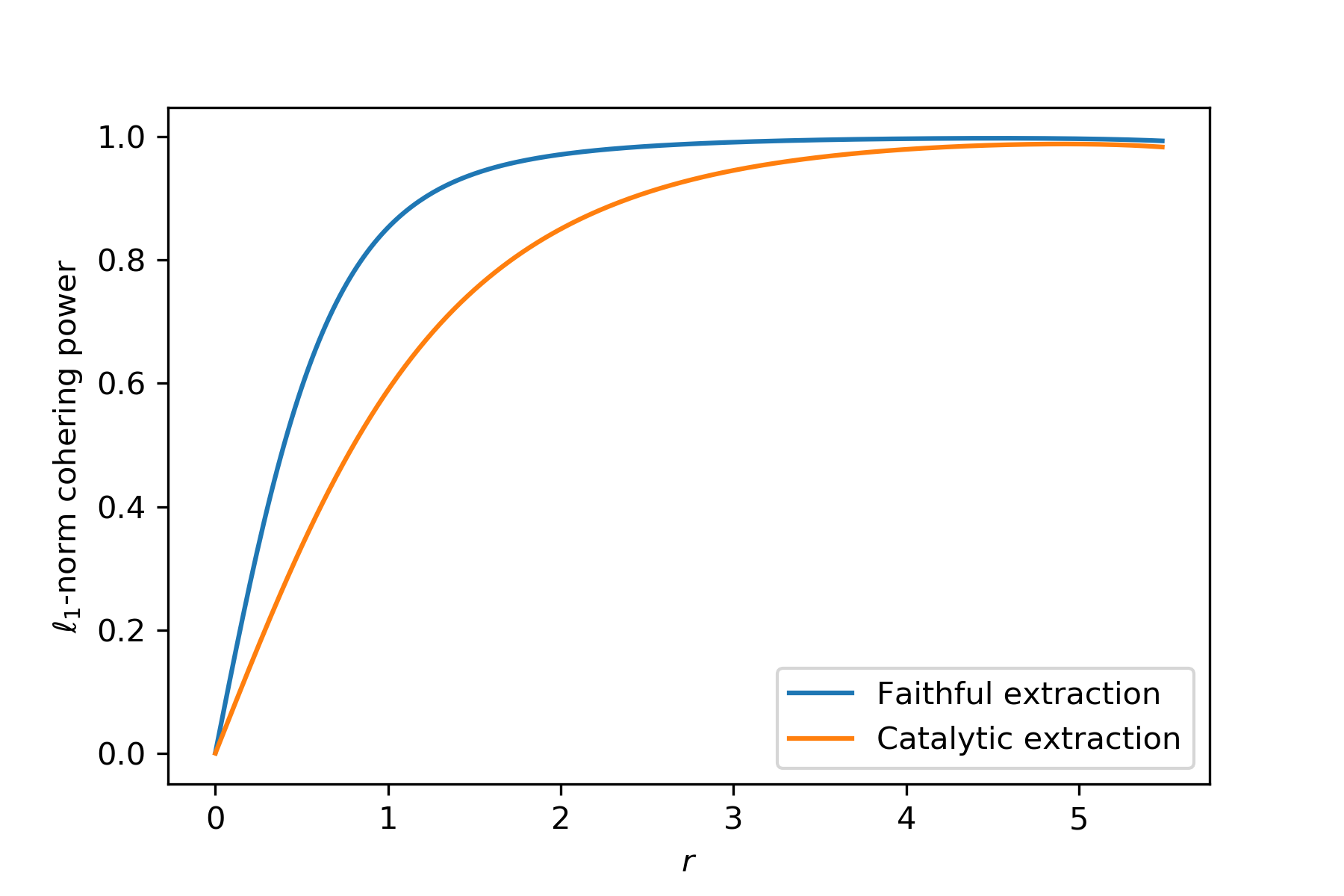}}
\caption{Cohering power of strong (upper curve) and catalytic (lower curve) extraction protocols for a reservoir in the coherent and vacuum squeezed states of a harmonic oscillator.}\label{fig5}
\end{figure}

As was also pointed out in \cite{PhysRevLett.113.150402}, the interaction given by \cref{eq28} for the weak protocol resembles closely that between a qubit and a single mode of the electromagnetic field given by the Jaynes-Cummings Hamiltonian \cite{1443594,doi:10.1080/09500349314551321}
\begin{equation}
    H=g\left(\begin{array}{cc}
         0& a^{\dagger} \\
         a& 0 
    \end{array}\right)
\end{equation}
where $a$, and $a^\dagger$ are the annihilation and creation operators of the field and $g$ is a coupling constant (for use of the Jaynes-Cummings interaction in coherence and catalysis see \cite{10.1088/1367-2630/ab7607}). It is interesting to note that the same kind of interaction also features in entanglement harvesting protocols \cite{PhysRevD.92.064042,PhysRevD.96.025020,PhysRevD.98.085007,PhysRevD.97.125002}. This raises the possibility of extraction of coherence from the vacuum state of a quantum field.

An open question is whether other faithful extraction protocols exist which could outperform the ones presented here. How to modify the protocols for reservoirs with energy levels of unequal distance, e.g. an atom, would also be of interest.
%%%%%%%%
\section*{acknowledgments}
The authors would like to thank C. Anastopoulos and R. N. Morty for helpful discussions in preparation of this manuscript. N.K.K. Acknowledges support by Grant No. E611 from the Research Committee of the University of Patras via the K. Karatheodoris program.
\vfill
\bibliography{extraction}

\begin{thebibliography}{39}%
\makeatletter
\providecommand \@ifxundefined [1]{%
 \@ifx{#1\undefined}
}%
\providecommand \@ifnum [1]{%
 \ifnum #1\expandafter \@firstoftwo
 \else \expandafter \@secondoftwo
 \fi
}%
\providecommand \@ifx [1]{%
 \ifx #1\expandafter \@firstoftwo
 \else \expandafter \@secondoftwo
 \fi
}%
\providecommand \natexlab [1]{#1}%
\providecommand \enquote  [1]{``#1''}%
\providecommand \bibnamefont  [1]{#1}%
\providecommand \bibfnamefont [1]{#1}%
\providecommand \citenamefont [1]{#1}%
\providecommand \href@noop [0]{\@secondoftwo}%
\providecommand \href [0]{\begingroup \@sanitize@url \@href}%
\providecommand \@href[1]{\@@startlink{#1}\@@href}%
\providecommand \@@href[1]{\endgroup#1\@@endlink}%
\providecommand \@sanitize@url [0]{\catcode `\\12\catcode `\$12\catcode
  `\&12\catcode `\#12\catcode `\^12\catcode `\_12\catcode `\%12\relax}%
\providecommand \@@startlink[1]{}%
\providecommand \@@endlink[0]{}%
\providecommand \url  [0]{\begingroup\@sanitize@url \@url }%
\providecommand \@url [1]{\endgroup\@href {#1}{\urlprefix }}%
\providecommand \urlprefix  [0]{URL }%
\providecommand \Eprint [0]{\href }%
\providecommand \doibase [0]{http://dx.doi.org/}%
\providecommand \selectlanguage [0]{\@gobble}%
\providecommand \bibinfo  [0]{\@secondoftwo}%
\providecommand \bibfield  [0]{\@secondoftwo}%
\providecommand \translation [1]{[#1]}%
\providecommand \BibitemOpen [0]{}%
\providecommand \bibitemStop [0]{}%
\providecommand \bibitemNoStop [0]{.\EOS\space}%
\providecommand \EOS [0]{\spacefactor3000\relax}%
\providecommand \BibitemShut  [1]{\csname bibitem#1\endcsname}%
\let\auto@bib@innerbib\@empty
%</preamble>
\bibitem [{\citenamefont {Aberg}(2006)}]{aberg2006quantifying}%
  \BibitemOpen
  \bibfield  {author} {\bibinfo {author} {\bibfnamefont {J.}~\bibnamefont
  {Aberg}},\ }\href@noop {} {\enquote {\bibinfo {title} {Quantifying
  superposition},}\ } (\bibinfo {year} {2006}),\ \Eprint
  {http://arxiv.org/abs/quant-ph/0612146} {arXiv:quant-ph/0612146 [quant-ph]}
  \BibitemShut {NoStop}%
\bibitem [{\citenamefont {Baumgratz}\ \emph {et~al.}(2014)\citenamefont
  {Baumgratz}, \citenamefont {Cramer},\ and\ \citenamefont
  {Plenio}}]{PhysRevLett.113.140401}%
  \BibitemOpen
  \bibfield  {author} {\bibinfo {author} {\bibfnamefont {T.}~\bibnamefont
  {Baumgratz}}, \bibinfo {author} {\bibfnamefont {M.}~\bibnamefont {Cramer}}, \
  and\ \bibinfo {author} {\bibfnamefont {M.~B.}\ \bibnamefont {Plenio}},\
  }\href {\doibase 10.1103/PhysRevLett.113.140401} {\bibfield  {journal}
  {\bibinfo  {journal} {Phys. Rev. Lett.}\ }\textbf {\bibinfo {volume} {113}},\
  \bibinfo {pages} {140401} (\bibinfo {year} {2014})}\BibitemShut {NoStop}%
\bibitem [{\citenamefont {Winter}\ and\ \citenamefont
  {Yang}(2016)}]{PhysRevLett.116.120404}%
  \BibitemOpen
  \bibfield  {author} {\bibinfo {author} {\bibfnamefont {A.}~\bibnamefont
  {Winter}}\ and\ \bibinfo {author} {\bibfnamefont {D.}~\bibnamefont {Yang}},\
  }\href {\doibase 10.1103/PhysRevLett.116.120404} {\bibfield  {journal}
  {\bibinfo  {journal} {Phys. Rev. Lett.}\ }\textbf {\bibinfo {volume} {116}},\
  \bibinfo {pages} {120404} (\bibinfo {year} {2016})}\BibitemShut {NoStop}%
\bibitem [{\citenamefont {Streltsov}\ \emph {et~al.}(2017)\citenamefont
  {Streltsov}, \citenamefont {Adesso},\ and\ \citenamefont
  {Plenio}}]{RevModPhys.89.041003}%
  \BibitemOpen
  \bibfield  {author} {\bibinfo {author} {\bibfnamefont {A.}~\bibnamefont
  {Streltsov}}, \bibinfo {author} {\bibfnamefont {G.}~\bibnamefont {Adesso}}, \
  and\ \bibinfo {author} {\bibfnamefont {M.~B.}\ \bibnamefont {Plenio}},\
  }\href {\doibase 10.1103/RevModPhys.89.041003} {\bibfield  {journal}
  {\bibinfo  {journal} {Rev. Mod. Phys.}\ }\textbf {\bibinfo {volume} {89}},\
  \bibinfo {pages} {041003} (\bibinfo {year} {2017})}\BibitemShut {NoStop}%
\bibitem [{\citenamefont {Bartlett}\ \emph
  {et~al.}(2007{\natexlab{a}})\citenamefont {Bartlett}, \citenamefont
  {Rudolph},\ and\ \citenamefont {Spekkens}}]{RevModPhys.79.555}%
  \BibitemOpen
  \bibfield  {author} {\bibinfo {author} {\bibfnamefont {S.~D.}\ \bibnamefont
  {Bartlett}}, \bibinfo {author} {\bibfnamefont {T.}~\bibnamefont {Rudolph}}, \
  and\ \bibinfo {author} {\bibfnamefont {R.~W.}\ \bibnamefont {Spekkens}},\
  }\href {\doibase 10.1103/RevModPhys.79.555} {\bibfield  {journal} {\bibinfo
  {journal} {Rev. Mod. Phys.}\ }\textbf {\bibinfo {volume} {79}},\ \bibinfo
  {pages} {555} (\bibinfo {year} {2007}{\natexlab{a}})}\BibitemShut {NoStop}%
\bibitem [{\citenamefont {Gour}\ and\ \citenamefont
  {Spekkens}(2008)}]{Gour_2008}%
  \BibitemOpen
  \bibfield  {author} {\bibinfo {author} {\bibfnamefont {G.}~\bibnamefont
  {Gour}}\ and\ \bibinfo {author} {\bibfnamefont {R.~W.}\ \bibnamefont
  {Spekkens}},\ }\href {\doibase 10.1088/1367-2630/10/3/033023} {\bibfield
  {journal} {\bibinfo  {journal} {New Journal of Physics}\ }\textbf {\bibinfo
  {volume} {10}},\ \bibinfo {pages} {033023} (\bibinfo {year}
  {2008})}\BibitemShut {NoStop}%
\bibitem [{\citenamefont {Marvian}\ and\ \citenamefont
  {Spekkens}(2014{\natexlab{a}})}]{Marvian2014}%
  \BibitemOpen
  \bibfield  {author} {\bibinfo {author} {\bibfnamefont {I.}~\bibnamefont
  {Marvian}}\ and\ \bibinfo {author} {\bibfnamefont {R.~W.}\ \bibnamefont
  {Spekkens}},\ }\href {https://doi.org/10.1038/ncomms4821} {\bibfield
  {journal} {\bibinfo  {journal} {Nature Communications}\ }\textbf {\bibinfo
  {volume} {5}},\ \bibinfo {pages} {3821 EP } (\bibinfo {year}
  {2014}{\natexlab{a}})},\ \bibinfo {note} {article}\BibitemShut {NoStop}%
\bibitem [{\citenamefont {Marvian}\ and\ \citenamefont
  {Spekkens}(2014{\natexlab{b}})}]{PhysRevA.90.062110}%
  \BibitemOpen
  \bibfield  {author} {\bibinfo {author} {\bibfnamefont {I.}~\bibnamefont
  {Marvian}}\ and\ \bibinfo {author} {\bibfnamefont {R.~W.}\ \bibnamefont
  {Spekkens}},\ }\href {\doibase 10.1103/PhysRevA.90.062110} {\bibfield
  {journal} {\bibinfo  {journal} {Phys. Rev. A}\ }\textbf {\bibinfo {volume}
  {90}},\ \bibinfo {pages} {062110} (\bibinfo {year}
  {2014}{\natexlab{b}})}\BibitemShut {NoStop}%
\bibitem [{\citenamefont {Marvian}\ and\ \citenamefont
  {Spekkens}(2016)}]{PhysRevA.94.052324}%
  \BibitemOpen
  \bibfield  {author} {\bibinfo {author} {\bibfnamefont {I.}~\bibnamefont
  {Marvian}}\ and\ \bibinfo {author} {\bibfnamefont {R.~W.}\ \bibnamefont
  {Spekkens}},\ }\href {\doibase 10.1103/PhysRevA.94.052324} {\bibfield
  {journal} {\bibinfo  {journal} {Phys. Rev. A}\ }\textbf {\bibinfo {volume}
  {94}},\ \bibinfo {pages} {052324} (\bibinfo {year} {2016})}\BibitemShut
  {NoStop}%
\bibitem [{\citenamefont {Lloyd}(2011)}]{Lloyd_2011}%
  \BibitemOpen
  \bibfield  {author} {\bibinfo {author} {\bibfnamefont {S.}~\bibnamefont
  {Lloyd}},\ }\href {\doibase 10.1088/1742-6596/302/1/012037} {\bibfield
  {journal} {\bibinfo  {journal} {Journal of Physics: Conference Series}\
  }\textbf {\bibinfo {volume} {302}},\ \bibinfo {pages} {012037} (\bibinfo
  {year} {2011})}\BibitemShut {NoStop}%
\bibitem [{\citenamefont {\AA{}berg}(2014)}]{PhysRevLett.113.150402}%
  \BibitemOpen
  \bibfield  {author} {\bibinfo {author} {\bibfnamefont {J.}~\bibnamefont
  {\AA{}berg}},\ }\href {\doibase 10.1103/PhysRevLett.113.150402} {\bibfield
  {journal} {\bibinfo  {journal} {Phys. Rev. Lett.}\ }\textbf {\bibinfo
  {volume} {113}},\ \bibinfo {pages} {150402} (\bibinfo {year}
  {2014})}\BibitemShut {NoStop}%
\bibitem [{\citenamefont {Lostaglio}\ \emph
  {et~al.}(2015{\natexlab{a}})\citenamefont {Lostaglio}, \citenamefont
  {Jennings},\ and\ \citenamefont {Rudolph}}]{Lostaglio2015}%
  \BibitemOpen
  \bibfield  {author} {\bibinfo {author} {\bibfnamefont {M.}~\bibnamefont
  {Lostaglio}}, \bibinfo {author} {\bibfnamefont {D.}~\bibnamefont {Jennings}},
  \ and\ \bibinfo {author} {\bibfnamefont {T.}~\bibnamefont {Rudolph}},\ }\href
  {https://doi.org/10.1038/ncomms7383} {\bibfield  {journal} {\bibinfo
  {journal} {Nature Communications}\ }\textbf {\bibinfo {volume} {6}},\
  \bibinfo {pages} {6383 EP } (\bibinfo {year} {2015}{\natexlab{a}})},\
  \bibinfo {note} {article}\BibitemShut {NoStop}%
\bibitem [{\citenamefont {\ifmmode \acute{C}\else
  \'{C}\fi{}wikli\ifmmode~\acute{n}\else \'{n}\fi{}ski}\ \emph
  {et~al.}(2015)\citenamefont {\ifmmode \acute{C}\else
  \'{C}\fi{}wikli\ifmmode~\acute{n}\else \'{n}\fi{}ski}, \citenamefont
  {Studzi\ifmmode~\acute{n}\else \'{n}\fi{}ski}, \citenamefont {Horodecki},\
  and\ \citenamefont {Oppenheim}}]{PhysRevLett.115.210403}%
  \BibitemOpen
  \bibfield  {author} {\bibinfo {author} {\bibfnamefont {P.}~\bibnamefont
  {\ifmmode \acute{C}\else \'{C}\fi{}wikli\ifmmode~\acute{n}\else
  \'{n}\fi{}ski}}, \bibinfo {author} {\bibfnamefont {M.}~\bibnamefont
  {Studzi\ifmmode~\acute{n}\else \'{n}\fi{}ski}}, \bibinfo {author}
  {\bibfnamefont {M.}~\bibnamefont {Horodecki}}, \ and\ \bibinfo {author}
  {\bibfnamefont {J.}~\bibnamefont {Oppenheim}},\ }\href {\doibase
  10.1103/PhysRevLett.115.210403} {\bibfield  {journal} {\bibinfo  {journal}
  {Phys. Rev. Lett.}\ }\textbf {\bibinfo {volume} {115}},\ \bibinfo {pages}
  {210403} (\bibinfo {year} {2015})}\BibitemShut {NoStop}%
\bibitem [{\citenamefont {Lostaglio}\ \emph
  {et~al.}(2015{\natexlab{b}})\citenamefont {Lostaglio}, \citenamefont
  {Korzekwa}, \citenamefont {Jennings},\ and\ \citenamefont
  {Rudolph}}]{PhysRevX.5.021001}%
  \BibitemOpen
  \bibfield  {author} {\bibinfo {author} {\bibfnamefont {M.}~\bibnamefont
  {Lostaglio}}, \bibinfo {author} {\bibfnamefont {K.}~\bibnamefont {Korzekwa}},
  \bibinfo {author} {\bibfnamefont {D.}~\bibnamefont {Jennings}}, \ and\
  \bibinfo {author} {\bibfnamefont {T.}~\bibnamefont {Rudolph}},\ }\href
  {\doibase 10.1103/PhysRevX.5.021001} {\bibfield  {journal} {\bibinfo
  {journal} {Phys. Rev. X}\ }\textbf {\bibinfo {volume} {5}},\ \bibinfo {pages}
  {021001} (\bibinfo {year} {2015}{\natexlab{b}})}\BibitemShut {NoStop}%
\bibitem [{\citenamefont {Narasimhachar}\ and\ \citenamefont
  {Gour}(2015)}]{Narasimhachar2015}%
  \BibitemOpen
  \bibfield  {author} {\bibinfo {author} {\bibfnamefont {V.}~\bibnamefont
  {Narasimhachar}}\ and\ \bibinfo {author} {\bibfnamefont {G.}~\bibnamefont
  {Gour}},\ }\href {https://doi.org/10.1038/ncomms8689} {\bibfield  {journal}
  {\bibinfo  {journal} {Nature Communications}\ }\textbf {\bibinfo {volume}
  {6}},\ \bibinfo {pages} {7689 EP } (\bibinfo {year} {2015})},\ \bibinfo
  {note} {article}\BibitemShut {NoStop}%
\bibitem [{\citenamefont {Korzekwa}\ \emph {et~al.}(2016)\citenamefont
  {Korzekwa}, \citenamefont {Lostaglio}, \citenamefont {Oppenheim},\ and\
  \citenamefont {Jennings}}]{Korzekwa_2016}%
  \BibitemOpen
  \bibfield  {author} {\bibinfo {author} {\bibfnamefont {K.}~\bibnamefont
  {Korzekwa}}, \bibinfo {author} {\bibfnamefont {M.}~\bibnamefont {Lostaglio}},
  \bibinfo {author} {\bibfnamefont {J.}~\bibnamefont {Oppenheim}}, \ and\
  \bibinfo {author} {\bibfnamefont {D.}~\bibnamefont {Jennings}},\ }\href
  {\doibase 10.1088/1367-2630/18/2/023045} {\bibfield  {journal} {\bibinfo
  {journal} {New Journal of Physics}\ }\textbf {\bibinfo {volume} {18}},\
  \bibinfo {pages} {023045} (\bibinfo {year} {2016})}\BibitemShut {NoStop}%
\bibitem [{\citenamefont {Streltsov}\ \emph {et~al.}(2015)\citenamefont
  {Streltsov}, \citenamefont {Singh}, \citenamefont {Dhar}, \citenamefont
  {Bera},\ and\ \citenamefont {Adesso}}]{PhysRevLett.115.020403}%
  \BibitemOpen
  \bibfield  {author} {\bibinfo {author} {\bibfnamefont {A.}~\bibnamefont
  {Streltsov}}, \bibinfo {author} {\bibfnamefont {U.}~\bibnamefont {Singh}},
  \bibinfo {author} {\bibfnamefont {H.~S.}\ \bibnamefont {Dhar}}, \bibinfo
  {author} {\bibfnamefont {M.~N.}\ \bibnamefont {Bera}}, \ and\ \bibinfo
  {author} {\bibfnamefont {G.}~\bibnamefont {Adesso}},\ }\href {\doibase
  10.1103/PhysRevLett.115.020403} {\bibfield  {journal} {\bibinfo  {journal}
  {Phys. Rev. Lett.}\ }\textbf {\bibinfo {volume} {115}},\ \bibinfo {pages}
  {020403} (\bibinfo {year} {2015})}\BibitemShut {NoStop}%
\bibitem [{\citenamefont {Chitambar}\ and\ \citenamefont
  {Hsieh}(2016)}]{PhysRevLett.117.020402}%
  \BibitemOpen
  \bibfield  {author} {\bibinfo {author} {\bibfnamefont {E.}~\bibnamefont
  {Chitambar}}\ and\ \bibinfo {author} {\bibfnamefont {M.-H.}\ \bibnamefont
  {Hsieh}},\ }\href {\doibase 10.1103/PhysRevLett.117.020402} {\bibfield
  {journal} {\bibinfo  {journal} {Phys. Rev. Lett.}\ }\textbf {\bibinfo
  {volume} {117}},\ \bibinfo {pages} {020402} (\bibinfo {year}
  {2016})}\BibitemShut {NoStop}%
\bibitem [{\citenamefont {Lostaglio}\ and\ \citenamefont
  {M\"uller}(2019)}]{PhysRevLett.123.020403}%
  \BibitemOpen
  \bibfield  {author} {\bibinfo {author} {\bibfnamefont {M.}~\bibnamefont
  {Lostaglio}}\ and\ \bibinfo {author} {\bibfnamefont {M.~P.}\ \bibnamefont
  {M\"uller}},\ }\href {\doibase 10.1103/PhysRevLett.123.020403} {\bibfield
  {journal} {\bibinfo  {journal} {Phys. Rev. Lett.}\ }\textbf {\bibinfo
  {volume} {123}},\ \bibinfo {pages} {020403} (\bibinfo {year}
  {2019})}\BibitemShut {NoStop}%
\bibitem [{\citenamefont {Marvian}\ and\ \citenamefont
  {Spekkens}(2019)}]{PhysRevLett.123.020404}%
  \BibitemOpen
  \bibfield  {author} {\bibinfo {author} {\bibfnamefont {I.}~\bibnamefont
  {Marvian}}\ and\ \bibinfo {author} {\bibfnamefont {R.~W.}\ \bibnamefont
  {Spekkens}},\ }\href {\doibase 10.1103/PhysRevLett.123.020404} {\bibfield
  {journal} {\bibinfo  {journal} {Phys. Rev. Lett.}\ }\textbf {\bibinfo
  {volume} {123}},\ \bibinfo {pages} {020404} (\bibinfo {year}
  {2019})}\BibitemShut {NoStop}%
\bibitem [{\citenamefont {Vaccaro}\ \emph {et~al.}(2018)\citenamefont
  {Vaccaro}, \citenamefont {Croke},\ and\ \citenamefont
  {Barnett}}]{Vaccaro_2018}%
  \BibitemOpen
  \bibfield  {author} {\bibinfo {author} {\bibfnamefont {J.~A.}\ \bibnamefont
  {Vaccaro}}, \bibinfo {author} {\bibfnamefont {S.}~\bibnamefont {Croke}}, \
  and\ \bibinfo {author} {\bibfnamefont {S.~M.}\ \bibnamefont {Barnett}},\
  }\href {\doibase 10.1088/1751-8121/aac112} {\bibfield  {journal} {\bibinfo
  {journal} {Journal of Physics A: Mathematical and Theoretical}\ }\textbf
  {\bibinfo {volume} {51}},\ \bibinfo {pages} {414008} (\bibinfo {year}
  {2018})}\BibitemShut {NoStop}%
\bibitem [{\citenamefont {Bartlett}\ \emph {et~al.}(2006)\citenamefont
  {Bartlett}, \citenamefont {Rudolph}, \citenamefont {Spekkens},\ and\
  \citenamefont {Turner}}]{Bartlett_2006}%
  \BibitemOpen
  \bibfield  {author} {\bibinfo {author} {\bibfnamefont {S.~D.}\ \bibnamefont
  {Bartlett}}, \bibinfo {author} {\bibfnamefont {T.}~\bibnamefont {Rudolph}},
  \bibinfo {author} {\bibfnamefont {R.~W.}\ \bibnamefont {Spekkens}}, \ and\
  \bibinfo {author} {\bibfnamefont {P.~S.}\ \bibnamefont {Turner}},\ }\href
  {\doibase 10.1088/1367-2630/8/4/058} {\bibfield  {journal} {\bibinfo
  {journal} {New Journal of Physics}\ }\textbf {\bibinfo {volume} {8}},\
  \bibinfo {pages} {58} (\bibinfo {year} {2006})}\BibitemShut {NoStop}%
\bibitem [{\citenamefont {Poulin}\ and\ \citenamefont
  {Yard}(2007)}]{Poulin_2007}%
  \BibitemOpen
  \bibfield  {author} {\bibinfo {author} {\bibfnamefont {D.}~\bibnamefont
  {Poulin}}\ and\ \bibinfo {author} {\bibfnamefont {J.}~\bibnamefont {Yard}},\
  }\href {\doibase 10.1088/1367-2630/9/5/156} {\bibfield  {journal} {\bibinfo
  {journal} {New Journal of Physics}\ }\textbf {\bibinfo {volume} {9}},\
  \bibinfo {pages} {156} (\bibinfo {year} {2007})}\BibitemShut {NoStop}%
\bibitem [{\citenamefont {Bartlett}\ \emph
  {et~al.}(2007{\natexlab{b}})\citenamefont {Bartlett}, \citenamefont
  {Rudolph}, \citenamefont {Sanders},\ and\ \citenamefont
  {Turner}}]{doi:10.1080/09500340701289254}%
  \BibitemOpen
  \bibfield  {author} {\bibinfo {author} {\bibfnamefont {S.~D.}\ \bibnamefont
  {Bartlett}}, \bibinfo {author} {\bibfnamefont {T.}~\bibnamefont {Rudolph}},
  \bibinfo {author} {\bibfnamefont {B.~C.}\ \bibnamefont {Sanders}}, \ and\
  \bibinfo {author} {\bibfnamefont {P.~S.}\ \bibnamefont {Turner}},\ }\href
  {\doibase 10.1080/09500340701289254} {\bibfield  {journal} {\bibinfo
  {journal} {Journal of Modern Optics}\ }\textbf {\bibinfo {volume} {54}},\
  \bibinfo {pages} {2211} (\bibinfo {year} {2007}{\natexlab{b}})}\BibitemShut
  {NoStop}%
\bibitem [{\citenamefont {Aharonov}\ \emph {et~al.}(1998)\citenamefont
  {Aharonov}, \citenamefont {Kaufherr}, \citenamefont {Popescu},\ and\
  \citenamefont {Reznik}}]{PhysRevLett.80.2023}%
  \BibitemOpen
  \bibfield  {author} {\bibinfo {author} {\bibfnamefont {Y.}~\bibnamefont
  {Aharonov}}, \bibinfo {author} {\bibfnamefont {T.}~\bibnamefont {Kaufherr}},
  \bibinfo {author} {\bibfnamefont {S.}~\bibnamefont {Popescu}}, \ and\
  \bibinfo {author} {\bibfnamefont {B.}~\bibnamefont {Reznik}},\ }\href
  {\doibase 10.1103/PhysRevLett.80.2023} {\bibfield  {journal} {\bibinfo
  {journal} {Phys. Rev. Lett.}\ }\textbf {\bibinfo {volume} {80}},\ \bibinfo
  {pages} {2023} (\bibinfo {year} {1998})}\BibitemShut {NoStop}%
\bibitem [{\citenamefont {Chitambar}\ and\ \citenamefont
  {Gour}(2019)}]{RevModPhys.91.025001}%
  \BibitemOpen
  \bibfield  {author} {\bibinfo {author} {\bibfnamefont {E.}~\bibnamefont
  {Chitambar}}\ and\ \bibinfo {author} {\bibfnamefont {G.}~\bibnamefont
  {Gour}},\ }\href {\doibase 10.1103/RevModPhys.91.025001} {\bibfield
  {journal} {\bibinfo  {journal} {Rev. Mod. Phys.}\ }\textbf {\bibinfo {volume}
  {91}},\ \bibinfo {pages} {025001} (\bibinfo {year} {2019})}\BibitemShut
  {NoStop}%
\bibitem [{\citenamefont {Mani}\ and\ \citenamefont
  {Karimipour}(2015)}]{PhysRevA.92.032331}%
  \BibitemOpen
  \bibfield  {author} {\bibinfo {author} {\bibfnamefont {A.}~\bibnamefont
  {Mani}}\ and\ \bibinfo {author} {\bibfnamefont {V.}~\bibnamefont
  {Karimipour}},\ }\href {\doibase 10.1103/PhysRevA.92.032331} {\bibfield
  {journal} {\bibinfo  {journal} {Phys. Rev. A}\ }\textbf {\bibinfo {volume}
  {92}},\ \bibinfo {pages} {032331} (\bibinfo {year} {2015})}\BibitemShut
  {NoStop}%
\bibitem [{\citenamefont {Bu}\ \emph {et~al.}(2017)\citenamefont {Bu},
  \citenamefont {Kumar}, \citenamefont {Zhang},\ and\ \citenamefont
  {Wu}}]{BU20171670}%
  \BibitemOpen
  \bibfield  {author} {\bibinfo {author} {\bibfnamefont {K.}~\bibnamefont
  {Bu}}, \bibinfo {author} {\bibfnamefont {A.}~\bibnamefont {Kumar}}, \bibinfo
  {author} {\bibfnamefont {L.}~\bibnamefont {Zhang}}, \ and\ \bibinfo {author}
  {\bibfnamefont {J.}~\bibnamefont {Wu}},\ }\href {\doibase
  https://doi.org/10.1016/j.physleta.2017.03.022} {\bibfield  {journal}
  {\bibinfo  {journal} {Physics Letters A}\ }\textbf {\bibinfo {volume}
  {381}},\ \bibinfo {pages} {1670 } (\bibinfo {year} {2017})}\BibitemShut
  {NoStop}%
\bibitem [{\citenamefont {{Chen}}\ \emph {et~al.}(2019)\citenamefont {{Chen}},
  \citenamefont {{Zhang}}, \citenamefont {{Zhou}},\ and\ \citenamefont
  {{Zhao}}}]{2019arXiv190609067C}%
  \BibitemOpen
  \bibfield  {author} {\bibinfo {author} {\bibfnamefont {S.}~\bibnamefont
  {{Chen}}}, \bibinfo {author} {\bibfnamefont {X.}~\bibnamefont {{Zhang}}},
  \bibinfo {author} {\bibfnamefont {Y.}~\bibnamefont {{Zhou}}}, \ and\ \bibinfo
  {author} {\bibfnamefont {Q.}~\bibnamefont {{Zhao}}},\ }\href@noop {}
  {\bibfield  {journal} {\bibinfo  {journal} {arXiv e-prints}\ ,\ \bibinfo
  {eid} {arXiv:1906.09067}} (\bibinfo {year} {2019})},\ \Eprint
  {http://arxiv.org/abs/1906.09067} {arXiv:1906.09067 [quant-ph]} \BibitemShut
  {NoStop}%
\bibitem [{\citenamefont {Glauber}(1963)}]{PhysRevLett.10.84}%
  \BibitemOpen
  \bibfield  {author} {\bibinfo {author} {\bibfnamefont {R.~J.}\ \bibnamefont
  {Glauber}},\ }\href {\doibase 10.1103/PhysRevLett.10.84} {\bibfield
  {journal} {\bibinfo  {journal} {Phys. Rev. Lett.}\ }\textbf {\bibinfo
  {volume} {10}},\ \bibinfo {pages} {84} (\bibinfo {year} {1963})}\BibitemShut
  {NoStop}%
\bibitem [{\citenamefont {Schnabel}(2017)}]{SCHNABEL20171}%
  \BibitemOpen
  \bibfield  {author} {\bibinfo {author} {\bibfnamefont {R.}~\bibnamefont
  {Schnabel}},\ }\href {\doibase https://doi.org/10.1016/j.physrep.2017.04.001}
  {\bibfield  {journal} {\bibinfo  {journal} {Physics Reports}\ }\textbf
  {\bibinfo {volume} {684}},\ \bibinfo {pages} {1 } (\bibinfo {year} {2017})},\
  \bibinfo {note} {squeezed states of light and their applications in laser
  interferometers}\BibitemShut {NoStop}%
\bibitem [{\citenamefont {{Jaynes}}\ and\ \citenamefont
  {{Cummings}}(1963)}]{1443594}%
  \BibitemOpen
  \bibfield  {author} {\bibinfo {author} {\bibfnamefont {E.~T.}\ \bibnamefont
  {{Jaynes}}}\ and\ \bibinfo {author} {\bibfnamefont {F.~W.}\ \bibnamefont
  {{Cummings}}},\ }\href {\doibase 10.1109/PROC.1963.1664} {\bibfield
  {journal} {\bibinfo  {journal} {Proceedings of the IEEE}\ }\textbf {\bibinfo
  {volume} {51}},\ \bibinfo {pages} {89} (\bibinfo {year} {1963})}\BibitemShut
  {NoStop}%
\bibitem [{\citenamefont {Shore}\ and\ \citenamefont
  {Knight}(1993)}]{doi:10.1080/09500349314551321}%
  \BibitemOpen
  \bibfield  {author} {\bibinfo {author} {\bibfnamefont {B.~W.}\ \bibnamefont
  {Shore}}\ and\ \bibinfo {author} {\bibfnamefont {P.~L.}\ \bibnamefont
  {Knight}},\ }\href {\doibase 10.1080/09500349314551321} {\bibfield  {journal}
  {\bibinfo  {journal} {Journal of Modern Optics}\ }\textbf {\bibinfo {volume}
  {40}},\ \bibinfo {pages} {1195} (\bibinfo {year} {1993})}\BibitemShut
  {NoStop}%
\bibitem [{\citenamefont {Messinger}\ \emph {et~al.}(2020)\citenamefont
  {Messinger}, \citenamefont {Ritboon}, \citenamefont {Crimin}, \citenamefont
  {Croke},\ and\ \citenamefont {Barnett}}]{10.1088/1367-2630/ab7607}%
  \BibitemOpen
  \bibfield  {author} {\bibinfo {author} {\bibfnamefont {A.}~\bibnamefont
  {Messinger}}, \bibinfo {author} {\bibfnamefont {A.}~\bibnamefont {Ritboon}},
  \bibinfo {author} {\bibfnamefont {F.}~\bibnamefont {Crimin}}, \bibinfo
  {author} {\bibfnamefont {S.}~\bibnamefont {Croke}}, \ and\ \bibinfo {author}
  {\bibfnamefont {S.}~\bibnamefont {Barnett}},\ }\href
  {http://iopscience.iop.org/10.1088/1367-2630/ab7607} {\bibfield  {journal}
  {\bibinfo  {journal} {New Journal of Physics}\ } (\bibinfo {year}
  {2020})}\BibitemShut {NoStop}%
\bibitem [{\citenamefont {Pozas-Kerstjens}\ and\ \citenamefont
  {Mart\'{\i}n-Mart\'{\i}nez}(2015)}]{PhysRevD.92.064042}%
  \BibitemOpen
  \bibfield  {author} {\bibinfo {author} {\bibfnamefont {A.}~\bibnamefont
  {Pozas-Kerstjens}}\ and\ \bibinfo {author} {\bibfnamefont {E.}~\bibnamefont
  {Mart\'{\i}n-Mart\'{\i}nez}},\ }\href {\doibase 10.1103/PhysRevD.92.064042}
  {\bibfield  {journal} {\bibinfo  {journal} {Phys. Rev. D}\ }\textbf {\bibinfo
  {volume} {92}},\ \bibinfo {pages} {064042} (\bibinfo {year}
  {2015})}\BibitemShut {NoStop}%
\bibitem [{\citenamefont {Simidzija}\ and\ \citenamefont
  {Mart\'{\i}n-Mart\'{\i}nez}(2017)}]{PhysRevD.96.025020}%
  \BibitemOpen
  \bibfield  {author} {\bibinfo {author} {\bibfnamefont {P.}~\bibnamefont
  {Simidzija}}\ and\ \bibinfo {author} {\bibfnamefont {E.}~\bibnamefont
  {Mart\'{\i}n-Mart\'{\i}nez}},\ }\href {\doibase 10.1103/PhysRevD.96.025020}
  {\bibfield  {journal} {\bibinfo  {journal} {Phys. Rev. D}\ }\textbf {\bibinfo
  {volume} {96}},\ \bibinfo {pages} {025020} (\bibinfo {year}
  {2017})}\BibitemShut {NoStop}%
\bibitem [{\citenamefont {Simidzija}\ and\ \citenamefont
  {Mart\'{\i}n-Mart\'{\i}nez}(2018)}]{PhysRevD.98.085007}%
  \BibitemOpen
  \bibfield  {author} {\bibinfo {author} {\bibfnamefont {P.}~\bibnamefont
  {Simidzija}}\ and\ \bibinfo {author} {\bibfnamefont {E.}~\bibnamefont
  {Mart\'{\i}n-Mart\'{\i}nez}},\ }\href {\doibase 10.1103/PhysRevD.98.085007}
  {\bibfield  {journal} {\bibinfo  {journal} {Phys. Rev. D}\ }\textbf {\bibinfo
  {volume} {98}},\ \bibinfo {pages} {085007} (\bibinfo {year}
  {2018})}\BibitemShut {NoStop}%
\bibitem [{\citenamefont {Simidzija}\ \emph {et~al.}(2018)\citenamefont
  {Simidzija}, \citenamefont {Jonsson},\ and\ \citenamefont
  {Mart\'{\i}n-Mart\'{\i}nez}}]{PhysRevD.97.125002}%
  \BibitemOpen
  \bibfield  {author} {\bibinfo {author} {\bibfnamefont {P.}~\bibnamefont
  {Simidzija}}, \bibinfo {author} {\bibfnamefont {R.~H.}\ \bibnamefont
  {Jonsson}}, \ and\ \bibinfo {author} {\bibfnamefont {E.}~\bibnamefont
  {Mart\'{\i}n-Mart\'{\i}nez}},\ }\href {\doibase 10.1103/PhysRevD.97.125002}
  {\bibfield  {journal} {\bibinfo  {journal} {Phys. Rev. D}\ }\textbf {\bibinfo
  {volume} {97}},\ \bibinfo {pages} {125002} (\bibinfo {year}
  {2018})}\BibitemShut {NoStop}%
\bibitem [{\citenamefont {Wong}\ and\ \citenamefont {Zhao}(2002)}]{Wong2002}%
  \BibitemOpen
  \bibfield  {author} {\bibinfo {author} {\bibnamefont {Wong}}\ and\ \bibinfo
  {author} {\bibnamefont {Zhao}},\ }\href {\doibase 10.1007/s00365-001-0019-3}
  {\bibfield  {journal} {\bibinfo  {journal} {Constructive Approximation}\
  }\textbf {\bibinfo {volume} {18}},\ \bibinfo {pages} {355} (\bibinfo {year}
  {2002})}\BibitemShut {NoStop}%
\end{thebibliography}%
\section{Supplemental material}
\subsection{Strong faithful extraction of quantum coherence stored in $d$-level systems}
\noindent Let
\begin{equation}\label{supeq1}
    {H}_S=\epsilon_0\sum_{i=0}^{d-1}i\ketbra{i}
\end{equation}
be the Hamiltonian of a system with $d$-energy levels that will act as a storage for coherence. In order to extract coherence from a larger system with Hamiltonian
\begin{equation}\label{supeq2}
    H_R=\epsilon_0\sum_{n=0}^{N-1}n\ketbra{n},
\end{equation}
we need to interact the combined system with the following unitary interaction
\begin{equation}\label{supeq3}
    V=\sum_{i,j=0}^{d-1}\ketbra{i}{j}\otimes\Delta^jP_d(\Delta^\dagger)^i,
\end{equation}
where 
\begin{equation}\label{supeq4}
    P_d=\sum_{n=0}^{N/d-1}\ketbra{nd}
\end{equation} is the projection onto the subspace spanned by those eigenstates of $H_R$ with energies some multiple of $d$, and
\begin{equation}\label{supeq5}
    \Delta=\sum_{n=0}^{N-2}\ketbra{n+1}{n}
\end{equation}
is the \emph{shift operator}. Since it can always be assumed that the reservoir is part of some larger system, we only consider the case in which the number of energy levels of the reservoir is also some multiple of $d$, $(N \mbox{ mod }d)=0$. From
\begin{equation}\label{supeq6}
P_d(\Delta^\dagger)^i\Delta^jP_d=\delta_{ij}P_d
\end{equation}
and
\begin{equation}\label{supeq7}
    \sum_{i=1}^{d-1}\Delta^iP_d(\Delta^\dagger)^i=I_R,
\end{equation}
it can be shown that $V$ is indeed unitary and also energy conserving, $[ H_{tot},V]=0$ where $H_{tot}=H_S+H_R$ is the total Hamiltonian of the combined system. Since for any $0\leq i,j\leq d-1$ and $0\leq n\leq N/d-1$
\begin{equation}\label{supeq8}
    V(\ket{i}\otimes\ket{dn+j})=\ket{j}\otimes\ket{dn+i},
\end{equation}
it follows that $\mathcal{C}(V)=0$ and $V\in\bar{\mathcal{L}}({H}_S+{H}_R)$, so the protocol is faithful.
\begin{figure}
    \centering
    \includegraphics[width=\columnwidth]{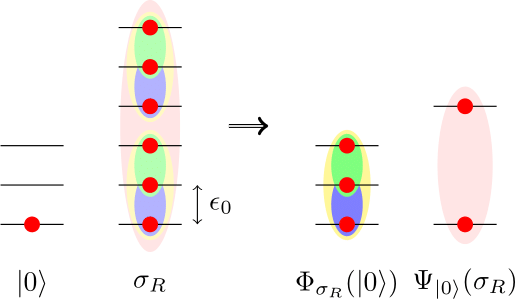}
    \caption{Strong faithful extraction from a $6$-level energy reservoir onto a $3$-level system initially in it's ground state. The protocol extracts the coherence between disjoint pairs of energy levels of the reservoir and stores it as coherence between corresponding energy levels in the system. After extraction any remaining coherence between energy levels with distance less than $3\epsilon_0$ in the reservoir has been destroyed. Note that the process cannot be repeated if we wish to extract further coherence to a $3$-level system, but can be repeated for extraction to $3$ qubits.}
    \label{fig6}
\end{figure}

Suppose that initially $\rho_S=\ketbra{0}$, evolving the combined system with the help of \cref{supeq3} and tracing out the reservoir we find that the state of the system is given by the action of an induced channel $\Phi_{\sigma_R}$ acting on $\mathscr{H}_S$
\begin{equation}\label{supeq9}
    \Phi_{\sigma_R}(\ket{0})=\sum_{i,j=0}^{d-1}\tr(\Delta^jP_d(\Delta^\dagger)^i\sigma_R)\ketbra{i}{j}.
\end{equation}
Similarly by tracing out $S$ we obtain the new state of the reservoir given by the action of an induced quantum channel $\Psi_{\ket{0}}$ acting on $\mathscr{H}_R$
\begin{equation}\label{supeq10}
    \Psi_{\ket{0}}(\sigma_R)=P_d\left(\sum_{i=0}^{d-1}(\Delta^\dagger)^i\sigma_R\Delta^i\right)P_d.
\end{equation}
From \cref{supeq9}, the amount of extracted coherence, as measured using the $\ell1$-norm of coherence, is equal to
\begin{equation}\label{supeq11}
C_{\ell_1}(\Phi_{\sigma_R}(\ket{0}))=2\sum_{j>i=0}^{d-1}\abs{\tr(\Delta^jP_d(\Delta^\dagger)^i\sigma_R)}.
\end{equation}
By expanding 
\begin{equation}\label{supeq12}
    \tr(\Delta^jP_d(\Delta^\dagger)^i\sigma_R)=\sum_{n=0}^{N/d-1}\sigma_{nd+i,nd+j}
\end{equation}
where $\sigma_{n,n'}=\bra{n}\sigma_R\ket{n'}$ are the reservoir's elements, we can observe that the protocol extracts the coherence between disjoint pairs of energy levels of the reservoir with labels $(n\mbox{ mod }d)=i$ and $(n\mbox{ mod }d)=j$ and stores it as coherence between the $i$-th and $j$-th energy level of the system (\cref{fig6}).

From \cref{supeq10} it follows that we can treat the reservoir as a system with an effective Hamiltonian equal to 
\begin{equation}\label{supeq13}
    H'_R=P_dH_RP_d.
\end{equation}
In order to extract coherence a second time we therefore need to scale the extracted system's Hamiltonian by a factor of $d$. Replacing, $P_d\to P_{d^2}$, $\Delta\to\Delta^dP_d$ and $\sigma_R\to\Psi_{\ket{0}}(\sigma_R)$ in \cref{supeq11} we can compute the newly extracted amount. Repeating the same kind of reasoning each time it can be shown by induction that after $m$ extractions, the state of the extracted system with Hamiltonian
\begin{equation}\label{supeq14}
    H_S^{(m)}=\bigoplus_{i=1}^{d^{m-1}}H_S
\end{equation}
will be equal to
\begin{equation}\label{supeq15}
    \rho_S^{(m)}=\sum_{i,j=0}^{d-1}\tr(\Delta^{jd^{m-1}}P_d^{(m)}(\Delta^\dagger)^{id^{m-1}}\sigma_R)\ketbra{i^{(m)}}{j^{(m)}}
\end{equation}
where $\ket{i^{(m)}}$ denote eigenstates of $H_S^{(m)}$ with energy equal to $id^{m-1}\epsilon_0$,
\begin{equation}\label{supeq16}
    P_d^{(m)}=\sum_{k=0}^{d^{m-1}-1}\Delta^kP_{d^m}(\Delta^\dagger)^k
\end{equation}
and 
\begin{equation}\label{supeq17}
    P_{d^m}=\sum_n\ketbra{nd^m}
\end{equation}
where the sum is taken over those integer values of $n\leq N/d^m-1$. Similarly the reservoir each time will be reduced to
\begin{equation}\label{supeq18}
    \sigma_R^{(m)}=P_{d^m}\left(\sum_{i=0}^{d^m-1}(\Delta^\dagger)^i\sigma_R\Delta^i\right)P_{d^m}
\end{equation}
and the amount of extracted coherence is equal to
\begin{equation}\label{supeq19}
    C_{\ell_1}(\rho_S^{(m)})=2\sum_{j>i=0}^{d-1}\abs{\tr(\Delta^{jd^{m-1}}P_d^{(m)}(\Delta^\dagger)^{id^{m-1}}\sigma_R)}.
\end{equation}
units of coherence. Since the total energy of the extracted systems cannot exceed that of the reservoir, $(N-1)\epsilon_0$, it follows that
\begin{equation}\label{supeq20}
    (d-1)\epsilon_0\sum_{m=1}^{M}d^{m-1}\leq (N-1)\epsilon_0
\end{equation}
so the protocol cannot be repeated more that {$M=\lfloor\log_dN\rfloor$} times.
%%%%%%%%%
\subsection{Asymptotic behaviour of $F$ and $G$}
\begin{lem}
Let
    \begin{equation}\label{supeq21}
    F_{d;k,k'}(x)=e^{-x}\sum_{n=0}^\infty \frac{x^{nd+\frac{k+k'}{2}}}{\sqrt{(nd+k)!(nd+k')!}}
    \end{equation}
with $0\leq k,k'\leq d-1$, then
    \begin{equation}\label{supeq22}
    \lim_{x\to\infty}F_{d;k,k'}(x)=\frac{1}{d}.
    \end{equation}
\end{lem}
\begin{proof}
Let $N$ be a sufficiently large integer and 
\begin{equation}\label{supeq23}
    c_n(d)=\frac{1}{\sqrt{(nd+k)!(nd+k')!}}
\end{equation} 
then
\begin{equation}\label{supeq24}
    e^xx^{-\frac{k+k'}{2}}F_{d;k,k'}(x)=\sum_{n=0}^Nc_n(d)x^{nd}+\sum_{n>N}c_n(d)x^{nd}.
\end{equation}
By logarithmic convexity of the gamma function it follows that for very large values of $n$
    \begin{equation}\label{supeq25}
    \sqrt{(nd+k)!(nd+k')!}\simeq\Gamma\left(nd+\frac{k+k'}{2}+1\right),
    \end{equation}
\cref{supeq24} can then be rewritten as
    \begin{align}
    &e^xx^{-\frac{k+k'}{2}}F_{d;k,k'}(x)\simeq E_{d,\frac{k+k'}{2}+1}(x^d)\nonumber\\
    &+\sum_{n=0}^Nx^{nd}\left(\frac{1}{\sqrt{(nd+k)!(nd+k')!}}-\frac{1}{\Gamma\left(nd+\frac{k+k'}{2}+1\right)}\right)\label{supeq26}
    \end{align}
where 
\begin{equation}\label{supeq27}
    E_{\alpha,\beta}(x)=\sum_{n=0}^\infty\frac{x^n}{\Gamma(an+b)}
\end{equation}
is the two parameter Mittag-Leffler function. Multiplying \cref{supeq26} by $e^{-x}x^{\frac{k+k'}{2}}$, taking the limit $x\to\infty$ and employing the asymptotic expansion of $E_{\alpha,\beta}(x)$ \cite{Wong2002}
    \begin{equation}\label{supeq28}
    E_{d,\frac{k+k'}{2}+1}(x^d)=\frac{1}{d}x^{-\frac{k+k'}{2}}e^x+\frac{1}{d}\sum_{s\geq 1}X_s^{-\frac{k+k'}{2}}e^{X_s}+O(x^{-d})
    \end{equation}
completes the proof.
\end{proof}
\begin{lem}
Let
\begin{equation}\label{supeq30}
    G_{d;k,k'}(x)=\sum_{n=0}^{\infty}\left(\frac{x}{4}\right)^{nd+\frac{k+k'}{2}}\frac{\sqrt{(2nd+2k)!(2nd+2k')!}}{(nd+k)!(nd+k')!}
\end{equation}
with $0\leq k,k'\leq d-1$, then
\begin{equation}\label{supeq29}
\lim_{x\to1}G_{d;k,k'}(x)\sqrt{1-x}=\frac{1}{d}.
\end{equation}
\end{lem}
\begin{proof}
Let $N$ be a sufficiently large integer and 
\begin{equation}\label{supeq31}
    c_n(d)=\frac{1}{2^{2nd+k+k'}}\frac{\sqrt{(2nd+2k)!(2nd+2k')!}}{(nd+k)!(nd+k')!}
\end{equation}
then
\begin{equation}\label{supeq32}
    x^{-\frac{k+k'}{2}}G_{d;k,k'}(x)=\sum_{n=0}^{N}c_n(d)x^{nd}+\sum_{n>N}^{\infty}c_n(d)x^{nd}.
\end{equation}
Employing Stirling's approximation 
\begin{equation}\label{supeq33}
c_n(d)\simeq\frac{1}{\sqrt{n\pi d}}\simeq\frac{(2n)!}{2^{2n}(n!)^2}\frac{1}{\sqrt{d}},\quad n>>1
\end{equation}
and \cref{supeq32} can be rewritten as
\begin{align}
    &x^{-\frac{k+k'}{2}}G_{d,k,k'}(x)\nonumber\\
    &\simeq \frac{1}{\sqrt{d}}\sum_{n=0}^{\infty}x^{nd}\frac{(2n)!}{2^{2n}(n!)^2}\nonumber+\sum_{n=0}^Nx^{nd}\left(c_n(d)-\frac{1}{\sqrt{d}}\frac{(2n)!}{2^{2n}(n!)^2}\right)\nonumber\\
    &=\frac{1}{\sqrt{d(1-x^d)}}+\sum_{n=0}^Nx^{nd}\left(c_n(d)-\frac{1}{\sqrt{d}}\frac{(2n)!}{2^{2n}(n!)^2}\right)
\end{align}
Multiplying each side by $x^{\frac{k+k'}{2}}\sqrt{1-x}$ and taking the limit $x\to1$ completes the proof.
\end{proof}
\end{document}